\def\fortech{1}
\pdfoutput=1

\documentclass[a4paper]{statsoc}

\usepackage{geometry}
\usepackage{amsmath}
\usepackage{comment}
\usepackage{graphicx}
\usepackage[caption=false,font=footnotesize]{subfig} 
\usepackage{multirow}
\usepackage{float}
\usepackage{url}
\usepackage[pdftex, pdfborder={0 0 0}, colorlinks=false]{hyperref}
\hypersetup{pdfborder={0 0 0}, colorlinks = true, linkcolor=black, citecolor=black, filecolor=black, urlcolor=black}
\usepackage{etoolbox}
\usepackage{bm}
\usepackage[ruled,vlined]{algorithm2e}
\newsavebox\mybox
\usepackage{thmtools, thm-restate}
\declaretheorem{theorem}
\allowdisplaybreaks

\newtheorem{example}{Example}[section]
\newtheorem{lemma}{Lemma}
\newtheorem{definition}{Definition}
\DeclareMathOperator{\diag}{diag}

\DeclareMathOperator{\randalg}{\mathcal{A}}
\DeclareMathOperator{\range}{range}
\DeclareMathOperator{\sensitivity}{\mathcal{S}}
\DeclareGraphicsExtensions{.pdf,.png}

\newbool{techreportbool}
\ifundef{\fortech}{
  \setbool{techreportbool}{false}
}{
  \setbool{techreportbool}{true}
}
\newcommand{\conferenceversion}[1]{\notbool{techreportbool}{#1}{}}
\newcommand{\techreport}[1]{\ifbool{techreportbool}{#1}{}}
\newcommand{\ConfOrTech}[2]{\notbool{techreportbool}{#1}{#2}}

\newcommand{\set}[1]{{\ensuremath{\left\{#1\right\}}}}
\newcommand{\eat}[1]{}
\newcommand{\neweat}[1]{}
\newcommand{\varchi}[1]{absolute difference {#1}} 

\newcommand{\astrut}{\rule[-.3\baselineskip]{0pt}{1.1em}}

\newcommand{\probvec}[1]{$\theta_0${#1}}

\newcommand{\thetaj}{\text{\probvec{}}[j]}
\newcommand{\myn}{(n_1+n_2)}
\newcommand{\whichn}[1]{\sqrt{\dfrac{n_{#1}}{\myn}}}
\newcommand{\zdj}{(\widetilde{T}[j]+\widetilde{S}[j])}
\newcommand{\whichzj}[2]{\left(\dfrac{{#2}[j]-n_{#1}\thetaj}{\sqrt{n_{#1}}}\right)}

\title{Revisiting Differentially Private Hypothesis Tests for Categorical Data}

\author[Y.\ Wang]{Yue~Wang}

\author[J.\ Lee]{Jaewoo~Lee}

\author[Y.\ Wang, J.\ Lee and D.\ Kifer]{Daniel~Kifer}

\address{Department of Computer Science and Engineering,
Penn State University,
USA}

\begin{document}

\maketitle
\begin{abstract}
In this paper, we consider methods for performing hypothesis tests on data protected by a statistical disclosure control technology known as differential privacy.
Previous approaches to differentially private hypothesis testing either perturbed the test statistic with random noise having large variance (and resulted in a significant loss of power) or added smaller amounts of noise directly to the data but failed to adjust the test in response to the added noise (resulting in biased, unreliable $p$-values). In this paper, we develop a variety of practical hypothesis tests that address these problems. Using a different asymptotic regime that is more suited to hypothesis testing with privacy, we show a modified equivalence between chi-squared tests and likelihood ratio tests. We then develop differentially private likelihood ratio and chi-squared tests for a variety of applications on tabular data (i.e., independence, sample proportions, and goodness-of-fit tests). \neweat{An open problem is whether new test statistics specialized to differential privacy could lead to further improvements. To aid in this search, we further propose a permutation-based testbed that can allow experimenters to empirically estimate the behavior of new test statistics for private hypothesis testing before fully working out their mathematical details (such as approximate null distributions).} Experimental evaluations on small and large datasets using a wide variety of privacy settings demonstrate the practicality and reliability of our methods. 

\end{abstract}

\section{Introduction}\label{sec:intro}
Hypothesis testing is an important aspect of statistical analysis and data mining. 
Because of the possibility that the results of an analysis could leak private information \citep{homer:genome}, various research communities such as statistics, official statistics, and computer science have studied how to incorporate statistical disclosure control to prevent such leakage.

A relatively recent development is a computer science privacy definition known as $\epsilon$-differential privacy \citep{dwork2006calibrating}. Statistical disclosure control (SDC) techniques that satisfy differential privacy possess a variety of appealing mathematical guarantees on the privacy of individuals in the data. 
For instance, the output produced by the SDC techniques is randomized and its probability distribution is barely affected by the inclusion of any individual's record in the data \citep{dwork2006calibrating}. Furthermore, if data records are independent, it limits any Bayesian inference about an individual: the ratio of posterior odds to prior odds is guaranteed to be bounded by $e^\epsilon$  \citep{pufferfish} where $\epsilon$ is a privacy parameter. This guarantee even holds if an attacker has access to all but one of the records in the data.

Recent years have seen a rapid development of differentially private SDC techniques for fitting various models to data. 
However, hypothesis testing is largely unexplored to the extent that existing methods can only be used reliably in rare circumstances. Our paper addresses these limitations, but first let us examine what the difficulties are.

Now let us consider problems raised by earlier applications of differential
privacy to hypothesis testing
\citep{johnson2013privacy,uhler2013privacy,yu2014scalable}. 




Suppose we collected the voter data in Fig.~\ref{tab:election2}. One may be
interested in determining if it provides statistical evidence that voting
behavior and gender are not independent.

In the classical (non-private) setting, this question is typically answered by
performing chi-squared ($\chi^2$) or likelihood ratio
tests~\citep{largesamplebook} of independence. This would involve computing the
$\chi^2$ or likelihood ratio test statistics
and taking advantage of theoretical results stating that for such $2\times 2$
tables, under the null hypothesis of independence, these test statistics are
asymptotically distributed as chi-squared random variables with 1 degree of
freedom ~\citep{largesamplebook}. The \emph{$p$-value} is simply the probability
that such a chi-squared random variable would be larger than the actual test
statistic. For example,  the likelihood ratio statistic for Fig.~\ref{tab:election2}a is $2.918$ with a corresponding $p$-value of $0.0876$, which is generally not considered strong enough to rule out independence. \textcolor{black}{Note that exact tests and permutation tests are other alternatives, but our approach to hypothesis testing with differential privacy extends the asymptotic approaches.}

\begin{example}[Input perturbation]\label{ex:inputpert}
  To achieve $\epsilon$-differential privacy,  \cite{johnson2013privacy} propose to add independent Laplace$(b)$ noise with density $f(x;b)=\frac{1}{2b} e^{-|x|/b}$ and $b=2/\epsilon$ to each cell of the table. As an example, when we added this noise to  Fig.~\ref{tab:election2}a, we obtained Fig.~\ref{tab:election2}b. The next step~\citep{johnson2013privacy} is to simply run the noisy table through off-the-shelf statistical software (which is unaware of this added noise). Intuitively, this seems a bit dangerous as the point of hypothesis testing is to determine how an analysis is affected by noise in the observed data. On the other hand, theoretical arguments \citep{johnson2013privacy} showed that test statistics computed from the noisy tables (in place of the original tables) still asymptotically have a chi-squared distribution with 1 degree of freedom. What happens in practice? The $p$-values produced by this method are extremely biased and will often lead to false conclusions. For example,  the likelihood ratio statistic computed from the noisy table in Fig.~\ref{tab:election2}b is equal to  $6.939$ and off-the-shelf software would return an estimated $p$-value of $0.0084$, which is often considered statistically significant and clearly contradicts the likelihood ratio test on the original data. \textcolor{black}{While this is just one example from one table, our experiments in Section \ref{subsec:exprely} empirically confirm this trend}. This mismatch between theoretical arguments and empirical results also points out the need for a more reliable privacy-preserving statistical theory. 
\end{example}



\begin{figure}
 \centering
 \subfloat[original table]{
  \centering
  \begin{tabular}{|l|cc|}
   \hline
   & vote & not vote\\
   \hline
   male & 238\astrut & 262 \\
   female & 265 & 235 \\
   \hline
  \end{tabular}
  \label{tab:election2a}
 }
 \hfil
 \subfloat[$0.2$-differential privacy]{
  \centering
  \begin{tabular}{|l|cc|}
   \hline
   & vote & not vote\\
   \hline
   male & 227.85\astrut & 279.24 \\
   female & 253.11 & 221.42 \\
   \hline
  \end{tabular}
  \label{tab:election2b}
 }
 \vskip -13pt
 \caption{\label{tab:election2}Tabulated Election Data} 
\end{figure}

\begin{example}[Output perturbation]\label{ex:outputpert}
The unreliability of the algorithm proposed by ~\cite{johnson2013privacy} was noticed by \cite{uhler2013privacy}, they proposed an output perturbation method \citep{uhler2013privacy,yu2014scalable}: compute the chi-squared statistic on the original data (here it is $2.916$), determine a quantity called the \emph{sensitivity} $S$ (the worst-case change in chi-squared values due to the alteration of one individual's data),  add Laplace$(b)$ noise with $b=S/\epsilon$, and then use a different asymptotic distribution for computing the $p$-value. For $2\times 2$ tables with $n=1,000$ (as in our example), the sensitivity $S$ is at least $500$ (achieved by the worst-case tables $\left(\begin{smallmatrix}1 & 0 \\0 & 999\end{smallmatrix}\right)$ and $\left(\begin{smallmatrix}1 & 1 \\0 & 998\end{smallmatrix}\right)$ with $\chi^2$ values $1000$ and $499.5$, respectively). This noise has standard deviation at least $500\sqrt{2}/\epsilon$. When such noise is added to the chi-squared statistic of the original table $(i.e., 2.916)$, it completely overwhelms the original value. As a result,  \cite{uhler2013privacy} and \cite{yu2014scalable} identified special cases where the amount of noise they need to add for privacy can be substantially reduced.
%
\end{example}

In this paper, similar to \cite{johnson2013privacy}, we first add noise to the input data before computing the test statistics. However, to get practical tests that work well on small and large datasets, we need to modify the $p$-value computation. First, we show how to appropriately adjust private statistical theory so that asymptotic results become a good approximation of what happens in practice. Then we derive the asymptotic distributions under this modified methodology for likelihood ratio and chi-squared tests for goodness-of-fit, sample proportions, and independence. We use these asymptotic distributions to produce $p$-values. Although more computationally expensive, we provide an extensive experimental evaluation on real datasets that shows our approach provides much more reliable $p$-values. We note that independent work by \cite{gaboardi2016diffchi} considers chi-squared tests under differential privacy. We elaborate on the differences in Section \ref{sec:relatedwork} (related work).

We introduce notations and terminologies in Section \ref{sec:preliminaries},
discuss related work in Section \ref{sec:relatedwork}, and present our various
statistical tests on privacy-enhanced tables in Section \ref{sec:asymptotics}.
Experiments appear in Section \ref{sec:experiments} and conclusions in Section \ref{sec:conc}. For completeness, proofs of our results appear in the online appendices.

\section{Preliminaries and Notations}\label{sec:preliminaries}
In this section, we introduce notations and review the necessary prerequisites.
\begin{list}{$\bullet$}{
\setlength{\itemsep}{0pt}
\setlength{\topsep}{3pt}
\setlength{\parsep}{3pt}
\setlength{\partopsep}{0pt}
\setlength{\leftmargin}{1em}}
\item Notations such as $T[\cdot]$ and $S[\cdot]$ indicate one-dimensional tables of counts. The $i^\text{th}$ entry of $T[\cdot]$ is denoted by $T[i]$.
\item Similarly, $T[\cdot,\cdot]$ is a two-dimensional table where $T[i,j]$ is the $(i,j)^\text{th}$ entry. We will use the standard shorthand $T[\bullet,j]=\sum_i T[i,j]$ and  $T[i,\bullet]=\sum_j T[i,j]$, as well as $T[\bullet,\bullet]=\sum_i\sum_j T[i,j]$.
\end{list}
The \textbf{size} of a table ($T[\cdot]$ or $T[\cdot,\cdot]$) is the sum of its counts.

\subsection{Review of Hypothesis Testing} \label{subsec:hypotest}
We consider the following types of statistical hypothesis tests:
\begin{list}{$\bullet$}{
\setlength{\itemsep}{0pt}
\setlength{\topsep}{3pt}
\setlength{\parsep}{3pt}
\setlength{\partopsep}{0pt}
\setlength{\leftmargin}{1em}}
\item \textbf{Goodness of fit}: given a table $T[\cdot]$ of size $n$ and a probability vector $\theta$, the null hypothesis is that $T[\cdot]$  was sampled from a Multinomial$(n,\theta)$ distribution.
\item \textbf{Sample proportions}: given tables $T[\cdot]$ of size $n_1$ and $S[\cdot]$ of size $n_2$, the null hypothesis is that they are samples from the same distribution. That is, $T[\cdot]\sim$Multinomial$(n_1,\theta)$ and $S[\cdot]\sim$Multinomial$(n_2,\theta)$, for some unknown $\theta$.
\item \textbf{Independence}: given a table $T[\cdot,\cdot]$ of size $n$, the null hypothesis is that the rows and columns are independent.
\end{list}
These hypotheses are commonly tested in the following ways:
\begin{definition}[Goodness of fit test]\label{def:goodness}
Compute either the \emph{likelihood ratio} statistic LR or chi-squared statistic $\chi^2$ as follows:

{\small
  \noindent\begin{minipage}{0.5\linewidth}
    \begin{equation}
      \label{eqn:lrone}
      LR = 2\sum_{i=1}^r T[i]\log\left(\frac{T[i]}{E[i]}\right)
    \end{equation}
  \end{minipage}%
  \begin{minipage}{0.5\linewidth}
    \begin{equation}
      \label{eqn:chione}
      \chi^2 = \sum_{i=1}^r \frac{(T[i]-E[i])^2}{E[i]},
    \end{equation}
  \end{minipage}\par\vspace{\belowdisplayskip}
}
where $E[i]=n\theta[i]$ are estimated expected null hypothesis cell counts \textcolor{black}{and
$r$ is the number of cells in $T$.} Under the null hypothesis, the asymptotic distribution of both $LR$ and $\chi^2$ is a chi-squared random variable with $r-1$ degrees of freedom. Thus the $p$-value can be approximated as the probability that the chi-squared random variable exceeds the chosen test statistic computed from $T$. Alternatively, we can sample many tables from  the Multinomial$(n,\theta)$ distribution and compute the test statistic of each one. The $p$-value could then be approximated as the fraction of sampled tables whose test statistic is greater than or equal to
 the test statistic of the actual table.
\end{definition}

\begin{definition}[Test of Sample Proportions]\label{def:homogeneity}
Given $T[\cdot]$ of size $n_1$ and $S[\cdot]$ of size $n_2$, compute  $LR$ or  $\chi^2$ as follows:
{\small
\begin{align}
LR &= 2\sum_{i=1}^r T[i]\log\left(\frac{T[i]}{E_1[i]}\right) + 2\sum_{i=1}^r S[i]\log\left(\frac{S[i]}{E_2[i]}\right)\label{eqn:lrtwo}\\
\chi^2 &= \sum_{i=1}^r \frac{(T[i]-E_1[i])^2}{E_1[i]} + \sum_{i=1}^r \frac{(S[i]-E_2[i])^2}{E_2[i]},\label{eqn:chitwo}
\end{align}}
where $E_1[i]=n_1(S[i]+T[i])/(n_1+n_2)$ and $E_2[i]=n_2(S[i]+T[i])/(n_1+n_2)$ are estimated expected null hypothesis cell counts, \textcolor{black}{and $r$ is the number of cells in $T$ and $S$}. Under the null hypothesis, the asymptotic distribution of both $LR$ and $\chi^2$ is a chi-squared random variable (r.v.) with $r-1$ degrees of freedom. We approximate the $p$-value as the probability that the chi-squared r.v. exceeds the chosen test statistic computed from $T$ and $S$.
\end{definition}

\begin{definition}[Test of Independence]\label{def:independence}
Given table $T[\cdot,\cdot]$ with $r$ rows and $c$ columns, compute LR or $\chi^2$ as:

{\small
  \noindent\begin{minipage}{0.5\linewidth}
    \begin{equation}
      LR = 2\sum_{i=1}^r\sum_{j=1}^c T[i,j]\log\left(\frac{T[i,j]}{E[i,j]}\right)\label{eqn:lrthree}
    \end{equation}
  \end{minipage}%
  \begin{minipage}{0.5\linewidth}
    \begin{equation}
      \chi^2 = \sum_{i=1}^r\sum_{j=1}^c \frac{(T[i,j]-E[i,j])^2}{E[i,j]},\label{eqn:chithree}
    \end{equation}
  \end{minipage}\par\vspace{\belowdisplayskip}
} 
where $E[i,j]=T[i,\bullet]T[\bullet,j]/T[\bullet,\bullet]$ are estimated expected null hypothesis cell counts. Under the null hypothesis, the asymptotic distribution of both $LR$ and $\chi^2$ is a chi-squared random variable with $(r-1)(c-1)$ degrees of freedom. The $p$-value can be approximated as the probability that the chi-squared random variable exceeds the chosen test statistic computed from $T$.
\end{definition}

A low $p$-value (say, 0.01, depending on the application) indicates strong evidence against the null hypothesis while a larger $p$-value indicates absence of evidence. As can be observed from Definitions \ref{def:goodness}, \ref{def:homogeneity}, \ref{def:independence}, the likelihood ratio and chi-squared tests are asymptotically equivalent \citep{largesamplebook}. However, as we explain in Section \ref{sec:asymptotics}, differences will emerge due to privacy constraints.

\subsection{Review of Differential Privacy}\label{subsec:dp}
Differential privacy \citep{dwork2006calibrating} is a set of restrictions on statistical disclosure control algorithms and guarantees that any data associated with an individual will have little impact on the result of the computation. 

\begin{definition}[Differential Privacy \citep{dwork2006calibrating}] \label{def:dp}
A randomized algorithm $\randalg$ satisfies $\epsilon$-differential privacy if for all contingency tables $T$ and $T^\prime$ that are derived from datasets that differ on the value of one record, and for  all $V \subseteq \range(\randalg)$,
$$P(\randalg(T) \in V) \leq e^\epsilon  P(\randalg(T^\prime) \in V)$$
\end{definition}

This definition means that if we modify any arbitrary record before tabulating our data into a table, the probability of generating any output is changed by a factor of at most $e^\epsilon$. When $\epsilon$ is small (i.e. close to $0$), the probabilities are barely affected. This severely limits the ability of an attacker to make inferences about any record in the data. For additional interpretations of the privacy guarantees of differential privacy and suggestions for setting the privacy parameter $\epsilon$, see \cite{dworkdiffp2006}, \cite{pufferfish} and \cite{designprivacy}.

There are many ways of constructing SDC algorithms that satisfy differential privacy. One of the simplest methods is called the \emph{Laplace Mechanism} \citep{dwork2006calibrating}. It relies on a concept called \emph{sensitivity} and adds noise scaled by the sensitivity value. 

\begin{definition}[Sensitivity \citep{dwork2006calibrating}] \label{def:sensitivity}
 Let $h$ be a function over contingency tables (the output of $h$ can be either a scalar or a vector). The \emph{sensitivity} of  $h$, denoted by $\sensitivity(h)$, is defined as  
 $\sensitivity(h)=\max_{T,T'} \|h(T)-h(T')\|_1$,
 where the maximum is over all pairs of tables that are derived from datasets differing on the value of one record.
\end{definition}

Intuitively, the sensitivity of a function $h$ measures the largest possible change that $h$ can experience as a result of modifying one record in some underlying dataset.

\begin{definition}[Laplace Mechanism \citep{dwork2006calibrating}]
Given a (vector or scalar valued) function $h$, privacy parameter $\epsilon$, contingency table $T$, and upper bound  $\mathcal{S}$ on the sensitivity $\sensitivity(h)$, the \emph{Laplace mechanism} adds independent Laplace($\mathcal{S}/\epsilon$) random variables (with density $f(x)=\frac{\epsilon}{2\mathcal{S}}e^{-\epsilon|x|/\mathcal{S}}$) to each component of $h(T)$. 
\end{definition}

\eat{
\techreport{ 
For example, Table~\ref{tab:unrestrictedneighbor} shows an example of two $n$-neighbors while Table~\ref{tab:neighbor} shows an example of two marginal-neighbors.

\begin{table}
\centering
\parbox{0.45\columnwidth}{
 \caption{Two contingency tables that are marginal-neighbors} 
 \label{tab:neighbor}
 \subfloat{
 \centering
  \begin{tabular}{|c|c|}
  \hline
  a & b \\
  \hline
  c & d \\
  \hline
  \end{tabular}
 }
 \hfil
 \subfloat{
 \centering
  \begin{tabular}{|c|c|}
  \hline
  a-1 & b+1 \\
  \hline
  c+1 & d-1 \\
  \hline
  \end{tabular}
 }
}
\hfil
\parbox{0.45\columnwidth}{
 \caption{Two contingency tables that are $n$-neighbors} 
 \label{tab:unrestrictedneighbor}
 \subfloat{
 \centering
  \begin{tabular}{|c|c|}
  \hline
  a & b \\
  \hline
  c & d \\
  \hline
  e & f \\
  \hline
  \end{tabular}
 }
 \hfil
 \subfloat{
 \centering
  \begin{tabular}{|c|c|}
  \hline
  a & b \\
  \hline
  c & d-1 \\
  \hline
  e+1 & f \\
  \hline
  \end{tabular}
 }
}
\end{table}
}
}

\section{Related Work}\label{sec:relatedwork}
Genome-wide association studies (GWAS) use statistical tests for finding
 associations between diseases and single-nucleotide polymorphisms (SNPs). The
 need for privacy became evident after \cite{homer:genome} raised the
 possibility of identifying individual participants in GWAS based on published
 SNP data. With GWAS in mind, \cite{johnson2013privacy} proposed differentially
 private algorithms for independence testing (e.g., $\chi^2$-tests) using input
 perturbation as discussed in Example \ref{ex:inputpert}. This method often
 produces unreliable conclusions except for extreme data sizes, as noted by
 ~\cite{uhler2013privacy} and our own experiments. Similar types of negative
 results produced by running off-the-shelf statistical analyses after input
 perturbation were reported by \cite{vu2009differential} and
 \cite{fienberg2010differential}. Thus, \cite{uhler2013privacy} instead proposed
 computing the true $\chi^2$ statistic, adding noise to it, and then adjusting
 the asymptotic distribution used to compute $p$-values.
Their method was limited to $3\times 2$ contingency tables where each of the 2
columns added up to $n/2$. \cite{yu2014scalable} later removed these
restrictions but still required the column-sums to be released exactly (i.e., in
a non-private way). In both cases, under the null hypothesis of independence,
collecting more data will not result in convergence to the non-private analysis
over the original data.
%
Independent work by \cite{gaboardi2016diffchi} follow earlier drafts of our work \citep{wang2015differentially}.
They consider chi-squared tests for goodness-of-fit and independence under the differential privacy model and also a weaker model known as approximate differential privacy  \citep{smoothsens,designprivacy}. There are several key distinctions between our work. First, we evaluate our methods on a variety of real datasets (they only consider synthetic data). Second, our formalization of a more accurate asymptotic regime provides asymptotic guarantees on the level of our tests. Without such a formalization, \cite{gaboardi2016diffchi} need synthetic data for empirical validation of Type 1 error (we provide such an empirical evaluation for our tests as well). We additionally consider the test of sample proportions and show a modified equivalence between chi-squared tests and likelihood ratio tests under the new asymptotic regime. Finally, our methods work for any $0$-mean finite variance noise distribution that is added to data (we specialize the discussion to Laplace noise, but any such distribution can be used in its place).

\neweat{
  (arXiv,
11/11/2015), which follow our work \citep{wang2015differentially}, consider only private chi-squared goodness-of-fit and independence
tests. Their 
results focus on the special case of Gaussian noise added to
data because of the special interactions
between this type of noise and the Central Limit Theorem (adding Gaussian noise can only satisfy approximate versions of differential privacy with weaker guarantees \citep{smoothsens,designprivacy}). They provide heuristics for Laplace noise but do not prove that their p-values can achieve a desired level of Type 1 error (even asymptotically). Our modified
asymptotics work for more general noise distributions (we only require $0$ mean and finite
variance) and so can handle differential privacy.
Furthermore, our modified privacy-aware equivalence between likelihood ratio and
chi-squared tests allows all of their results to be translated to likelihood ratio
tests as well. We also propose a permutation-based experimental testbed that
allows us to evaluate various test statistics in the context of privacy before computing their asymptotic
distributions.
}

In a  general setting, \cite{smith2011privacy} studied statistical
estimators that are known to have asymptotically normal distributions and
provided differentially-private versions of those estimators that are also
asymptotically normal. As with the tests proposed by \cite{johnson2013privacy}
(discussed in Example \ref{ex:inputpert}), very large data sizes are needed to
observe approximate normality.
%
Differential privacy has also been applied to other statistical tasks such as
computing commonly used robust statistical estimators
\citep{dwork2009differential} and computing private M-estimators
\citep{lei2011differentially}. \cite{chaudhuri2012convergence} established a
formal connection between differential privacy and robust statistics by deriving
convergence rates in terms of a concept called gross error sensitivity.
\cite{dwork2015private} presented a framework for controlling the false
discovery rate of a large sequence of hypothesis tests. The method relies on
injecting noise directly into $p$-values, which removes their guarantee that
they must be (approximately) uniformly distributed under the null hypothesis.
Specific instantiations of the framework for various statistical tests are not
given and its empirical performance is unknown. \cite{wassermanmini} studied
rates of convergence between true distributions and differentially private
estimates of distributions. They found that the provable convergence rates under
differential privacy were often slower than in the non-private case.

\section{Private Hypothesis Testing}\label{sec:asymptotics}
We consider the following setting: a data owner has a table of counts (e.g.,
$T[\cdot]$ or $T[\cdot,\cdot]$) and obtains noisy tables (e.g.,
$\widetilde{T}[\cdot]$, $\widetilde{T}[\cdot,\cdot]$) by adding independent
$0$-mean noise with finite variance to each table cell. The table size $n_0$, the
density function of the noise, and the noisy tables themselves are publicly
released. If the noise follows a Laplace$(2/\epsilon)$ distribution, then this
output satisfies $\epsilon$-differential privacy. Our goal is to conduct
goodness-of-fit, sample proportions, and independence tests using such noisy
data. We feel this is a natural setting as such releases of noisy counts do not
force the end-users into any particular analysis task.
We first justify our chosen asymptotic regime (Section \ref{subsec:asymptoticregime}), present the modified equivalence between chi-squared and likelihood ratio tests when computed over noisy data (Section \ref{subsec:chilr}) then derive asymptotic distributions for our tests (Section \ref{subsec:deriveasymp}). We use these distributions for $p$-value computation in Section \ref{subsec:pvaluealg} and then evaluate empirical performance in Section \ref{sec:experiments}.

\subsection{The Asymptotic Regime}\label{subsec:asymptoticregime}
\textcolor{black}{When discussing asymptotics, it is important to distinguish between the actual table that was collected, let us call it $T_0$, with sample size $n_0$ and the hypothetical data $T$ with sample size $n$ that goes to infinity.}
The key to asymptotically approximating the null distribution of a test statistic in the classical (non-private) case is the Central Limit Theorem (CLT), which, for example, states that as $n\rightarrow\infty$,  $\frac{T[1]-np_{1}}{\sqrt{n}}\rightarrow N(0,p_{1}(1-p_{1}))$ in distribution (where $np_{1}$ is the expected value of $T[1]$ and $N(0,\sigma^2)$ is the zero-mean Gaussian with variance $\sigma^2$). In practice, this Gaussian approximation works well even if $n$ is not large.
\textcolor{black}{In the private setting, the data collector is planning to release $T_0 + V_\epsilon$ where, for example, $V_\epsilon$ is a table of independent Laplace$(2/\epsilon)$ random variables.
One way to analyze it asymptotically is to replace $T_0$ with $T$ (and then let $n\rightarrow\infty$):}
{\small
\begin{align}
\hspace{-0.2cm}\frac{\widetilde{T}[1]-np_{1}}{\sqrt{n}} &= \frac{T[1]-np_{1}}{\sqrt{n}} + \frac{V_\epsilon[1]}{\sqrt{n}}\label{eqn:cltintuition}
\end{align}
}
The first term on the right hand side is often well-approximated by the  Gaussian distribution. As $n\rightarrow \infty$, the second term converges to $0$ in probability. However, we should not use $0$ as a finite-sample approximation to this term -- it is only accurate when $n$ is very large (in particular, $\sqrt{n}$ must be very large compared to the standard deviation of $V_\epsilon[1]$). 
 As discussed in Section \ref{sec:intro}, for many data sets of interest, this is not the case and so the extra noise due to SDC cannot be ignored.

Our proposed solution is based on the following idea. We view the first term in Equation \ref{eqn:cltintuition} as the signal and the second term as the noise. As $n\rightarrow\infty$, we want to maintain the \textcolor{black}{ratio of variance in the first term vs. variance in the second term} as in the actual data \textcolor{black}{$T_0$}. Thus we tie the standard deviation of the added noise to $\sqrt{n}$ as follows:
{\small
\begin{align}
\widetilde{T}[1]=T[1] + V_\epsilon[1]\kappa\sqrt{n}\label{eqn:newregime}
\end{align}
}
where $\kappa$ is a constant equal to $1/\sqrt{\textcolor{black}{n_0}}$. In this case, $\lim\limits_{n\rightarrow\infty} (\widetilde{T}[1]-np_1)/\sqrt{n} = N(0, p_1(1-p_1)) + V_\epsilon[1]\kappa$ so that asymptotics is only used to smooth out irregularities in the data and not to wipe out the noise terms. 
\textcolor{black}{
We justify this asymptotic regime with the following result, which states that if the data are well-approximated by a Gaussian, then the noisy data are just as well approximated by the limit of our proposed asymptotic regime. Here we measure the quality of the approximation by the largest difference in cumulative distribution functions ( the same measure used to quantify convergence to the Gaussian in the Central Limit Theorem). The proof is in online Appendix \ref{app:asympguarantee}.
\begin{restatable}{theorem}{thmasympguarantee}\label{thm:asympguarantee}
Let $X_{n_0}, X_{n_0+1}, X_{n_0+2}\dots$ be a sequence of (vector-valued) random variables such that $X_n/\sqrt{n}\rightarrow N(0,\sigma^2)$ in distribution. Let $Y$ be an independent random variable and let $Z_{n} = X_{n} + \frac{\sqrt{n}}{\sqrt{n_0}}Y$ for $n=n_0,n_{0}+1,n_0+2,\dots$. Then:
\begin{enumerate}
\item as $n\rightarrow\infty$, $Z_n/\sqrt{n}$ converges in distribution to a random variable, whose cumulative distribution function we will refer to as $G_Z$.
\item Letting $\Phi$ and $\phi$ represent the CDF and density of $N(0,\sigma^2)$ and letting $F_0$, $G_0$ represent the cumulative distribution functions of $X_{n_0}/\sqrt{n_0}$, $Z_{n_0}/\sqrt{n_0}$, respectively, then $\sup_{\vec{x}}|G_0(\vec{x}) - G_Z(\vec{x})| \leq \sup_{\vec{x}}|F_0(\vec{x}) - \Phi(\vec{x})|$.
\end{enumerate}
\end{restatable}
}  


\subsection{Relations between Likelihood Ratio and Chi-Squared Statistics}\label{subsec:chilr}
In classical statistics, likelihood ratio tests and chi-squared tests are asymptotically equivalent \citep{largesamplebook}. In the privacy-preserving case, this equivalence is modified.
\begin{restatable}{theorem}{thmchilr}\label{thm:chilr}
Suppose the probabilities under the true null hypothesis are nonzero.
Consider the noisy table $\widetilde{T}=T + V\kappa\sqrt{n}$ where $V$ is a $0$-mean random variable with fixed variance. Let $\widetilde{\chi^2}$ denote the chi-squared statistic obtained by replacing $T$ with the noisy $\widetilde{T}$ (and the expected counts $E$ computed from $\widetilde{T}$ instead of $T$). Let $\widetilde{LR}$ denote the likelihood ratio statistic with the same substitutions and from each term $i$ subtract $2(\widetilde{T}[i]-E[i])$. Then $\widetilde{\chi^2}$ and $\widetilde{LR}$ have the same asymptotic distribution as $n\rightarrow\infty$.
\end{restatable}
For proof, see \ConfOrTech{the full version of our paper
  \cite{wang2015differentially}}{the online Appendix \ref{app:chilr}}.
This theorem extends in the obvious way (and essentially the same proof) for the test of sample proportions for tables $S[\cdot]$ and $T[\cdot]$. \textcolor{black}{In the classical chi-squared and likelihood ratio tests, the tests can be inaccurate when some counts are very small (e.g., 0). Since we must use noisy counts instead of raw counts, we can extend the classical rules of thumb to the following: the tests can be used if all noisy cell counts are larger than 5 + several standard deviations (of the noise). More refined rules of thumb are an open problem.}
%
 The usefulness of this theorem is that sometimes it is easier to work with the likelihood ratio statistic and sometimes it is easier to work with $\chi^2$ when deriving asymptotic distributions.

\subsection{The Asymptotic Distributions}\label{subsec:deriveasymp}
In this section we derive asymptotic distributions for our various tests so that we can evaluate them in Section \ref{sec:experiments}. We phrase the results in terms of added Laplace noise for differential privacy, but the results hold when the noise $V$ follows any $0$-mean distribution with finite variance. We use these distributions for $p$-value computation in Section \ref{subsec:pvaluealg}. The proof of the following results are in \ConfOrTech{the full version of our paper
  \cite{wang2015differentially}}{the online Appendices \ref{app:indep} and \ref{app:homo}}.
\begin{restatable}{theorem}{thmindep}\textbf{(Independence testing).}
\label{thm:indep}
Let $T[\cdot,\cdot]$ be a contingency table sampled from a
Multinomial$(n_0,\theta_0)$ distribution. Consider the noisy table
$\widetilde{T}=T + V_\epsilon\kappa\sqrt{n_0}$ where  $V_\epsilon$ is a table 
of independent Laplace$(2/\epsilon)$ random variables. If the
rows and columns under $\theta_0$ are independent and if no cells have probability $0$,
then as $n_0\rightarrow\infty$, the
chi-squared statistic and the likelihood ratio statistic (Definition~\ref{def:independence}) computed from
$\widetilde{T}$ (instead of $T$) asymptotically have the distribution of the random variable:
{\small
\begin{equation*}
 \sum\limits_{ij}\frac{\left(A[i,j]+\kappa V^*[i,j]\right)^2}{\theta_0[i,j]}
  -\sum\limits_{i}\frac{\left(A[i,\bullet] + \kappa V^*[i,\bullet]\right)^2}{\theta_0[i,\bullet]}
  -\sum\limits_{j}\frac{\left(A[\bullet,j]+\kappa V^*[\bullet,j]\right)^2}{\theta_0[\bullet,j]} 
  + \frac{\left(A[\bullet,\bullet] + \kappa V^*[\bullet,\bullet]\right)^2}{1}
\end{equation*}
}
where $V^*$ has the same distribution as $V_\epsilon$ and the vectorized version $vec(A) \sim N(\bm{0}, \diag(vec(\theta_0)) - vec(\theta_0)vec(\theta_0)^t)$. It is asymptotically equivalent to the quantity we get
by replacing $\theta_0[i,j]$ with $\frac{\widetilde{T}[i,\bullet]\widetilde{T}[\bullet,j]}{\widetilde{T}[\bullet,\bullet]^2}$.
\end{restatable}

\begin{restatable}{theorem}{thmhomo}
\label{thm:homo}\textbf{(Test of Sample Proportions).}
Let $T[\cdot]$ and $S[\cdot]$ be samples from Multinomial$(n_1,\theta_0)$
and Multinomial$(n_2,\theta_0)$
distributions, respectively. Consider the noisy versions $\widetilde{T}=T+V^1_\epsilon\kappa_1\sqrt{n_1}$ and
$\widetilde{S}=S+V^2_\epsilon\kappa_2\sqrt{n_2}$ where $V^1_\epsilon$, $V^2_\epsilon$ are vectors of independent Laplace$(2/\epsilon)$ random variables. If no cells have
probability $0$, then as $n_1,n_2\rightarrow\infty$, the chi-squared and
likelihood ratio statistics (Definition~\ref{def:homogeneity}) computed from
$\widetilde{T}$ and $\widetilde{S}$ (instead of $T$ and $S$) asymptotically have the distribution of the random variable:
{\small
\begin{equation*}
  \sum\limits_j\left[ \sqrt{\frac{n_2}{n_1+n_2}}\left( A_1[j]+\kappa_1 V^*_1[j] \right) - \sqrt{\frac{n_1}{n_1+n_2}}\left( A_2[j]+\kappa_2 V^*_2[j] \right)\right]^2 \bigg/ \theta_0[j]
\end{equation*}
}
where $V^*_1,V^*_2$ are independent with the same distribution as $V^1_\epsilon,V^2_\epsilon$ and $A_1,A_2\sim N(\bm{0}, \diag(\theta_0) - \theta_0\theta_0^t)$, and $A_1$,
$A_2$ are independent. It is asymptotically equivalent to the quantity we get by
replacing $\theta_0[j]$ with
$(\widetilde{T}[j]+\widetilde{S}[j])/(n_1+n_2)$.
\end{restatable}

\subsection{p-Value Algorithms}\label{subsec:pvaluealg}
To apply Theorems \ref{thm:indep}, \ref{thm:homo}, \ref{thm:good} \textcolor{black}{(\ref{thm:good} comes later in this Section)}, we set
$\kappa=1/\sqrt{n_0}$, where $n_0$ is the actual table size (in the case of the test
of sample proportions, we set $\kappa_1=1/\sqrt{n_1}$ and $\kappa_2=1/\sqrt{n_2}$).
%
The recipe for computing $p$-values is relatively simple:
\begin{list}{$\bullet$}{
\setlength{\itemsep}{0pt}
\setlength{\topsep}{2pt}
\setlength{\parsep}{2pt}
\setlength{\partopsep}{0pt}
\setlength{\leftmargin}{1em}}
\item[1.] Compute the appropriate test statistic from the noisy tables (e.g., $\widetilde{T}[\cdot,\cdot]$, $\widetilde{T}[\cdot]$ and/or $\widetilde{S}[\cdot]$). Let $t^*$ be its value.
 \begin{list}{$\bullet$}{
\setlength{\itemsep}{0pt}
\setlength{\topsep}{1pt}
\setlength{\parsep}{1pt}
\setlength{\partopsep}{0pt}
\setlength{\leftmargin}{1em}}
\item For goodness of fit, the test statistics are given by Equations \ref{eqn:noisygoodchi} or \ref{eqn:noisygoodlr}.
\item For sample proportions, they are obtained from Definition \ref{def:homogeneity} by replacing the true tables $T[\cdot]$ and $S[\cdot]$ with their noisy versions $\widetilde{T}[\cdot]$ and $\widetilde{S}[\cdot]$. The $E_1$ and $E_2$ values from the definition must also be computed from the noisy tables.
\item For independence, they are obtained from Definition \ref{def:independence} by replacing the true table $T[\cdot,\cdot]$ with the noisy version $\widetilde{T}[\cdot,\cdot]$ and the $E[i,j]$ values from the definition are computed as $\widetilde{T}[i,\bullet]\widetilde{T}[\bullet,j]/\widetilde{T}[\bullet,\bullet]$.
\end{list}
\item[2.] Sample $m$ reference points $t_1,\dots, t_m$ (discussed next).
\item[3.] Set the $p$-value to be $\left|\set{t_i~:~ t_i\geq t^*}\right|/m$.
\end{list}

\setlength{\algoheightrule}{0pt}
\setlength{\algotitleheightrule}{0pt}
\begin{algorithm}[!t]
 \LinesNumbered
 \DontPrintSemicolon
 \SetKwInOut{Input}{input}
 \SetKwInOut{Output}{output}
 \begin{lrbox}{\mybox}
 \begin{minipage}{\hsize}
 \caption{Sampling for Independence Test} \label{alg:sampleindep}

 \Input{Noisy table $\widetilde{T}[\cdot,\cdot]$, $\epsilon>0$.}
  
 \For{\textbf{each} $i,j$}{
         $\theta_0[i,j] \leftarrow \frac{\widetilde{T}[i,\cdot]\widetilde{T}[\cdot,j]}{\widetilde{T}[\cdot,\cdot]^2}$\;
 }
 \For{$\ell=1,\dots, m$}{
 $vec(A) \sim N(\bm{0}, \diag(vec(\theta_0))-vec(\theta_0)vec(\theta_0)^{t})$\;
 Reshape $A$ to dimensions of $\widetilde{T}$\;
 $V^*[\cdot,\cdot] \sim Laplace(2/\epsilon)~~~~//$ \textit{fresh noise}\;\label{line:sens}

 $X \leftarrow A + V^*/ \sqrt{n_0}$\;\label{line:n}
 
 generate $t_\ell=\sum\limits_{ij}\frac{X[i,j]^2}{\theta_0[i,j]} - 
       \sum\limits_{i}\frac{X[i,\cdot]^2}{\theta_0[i,\cdot]} -
       \sum\limits_{j}\frac{X[\cdot,j]^2}{\theta_0[\cdot,j]} + 
       \frac{X[\cdot,\cdot]^2}{\theta_0[\cdot,\cdot]}$\;\label{line:sample}
 }
 \end{minipage}%
 \end{lrbox}
 \hspace*{-10pt}\framebox[\columnwidth]{\hspace*{15pt}\usebox\mybox\par}
\end{algorithm}
\begin{algorithm}[!t]
 \LinesNumbered
 \DontPrintSemicolon
 \SetKwInOut{Input}{input}
 \SetKwInOut{Output}{output}
 \begin{lrbox}{\mybox}
 \begin{minipage}{\hsize}
 \caption{Sampling for Test of Sample Proportions} \label{alg:samplehomo}

 \Input{Noisy tables $\widetilde{T}[\cdot]$ and $\widetilde{S}[\cdot]$, $\epsilon>0$.}
  
 \For{\textbf{each} $i$}{
         $\theta_0[i] \leftarrow (\widetilde{T}[i] + \widetilde{S}[i])/(n_1+n_2)$\;
 }
 \For{$\ell=1,\dots, m$}{
 $A_1, A_2 \sim N(\bm{0}, \diag(\theta_0)-\theta_0\theta_0^{t})$\;
 $V_1^*[\cdot], V_2^*[\cdot] \sim Laplace(2/\epsilon)~~~~//$ \textit{fresh noise}\;\label{line:sens2}

 $X_1 \leftarrow A_1 + V_1^*/ \sqrt{n_1}; X_2\leftarrow A_2 + V^*_2/\sqrt{n_2}$\;\label{line:n2}
 
 generate $t_\ell=\sum_j \left(\sqrt{\frac{n_2}{n_1+n_2}}X_1[j] - \sqrt{\frac{n_1}{n_1+n_2}}X_2[j]\right)^2 \bigg/\theta_0[j]$\;\label{line:sample2}
 }
 \end{minipage}%
 \end{lrbox}
 \hspace*{-10pt}\framebox[\columnwidth]{\hspace*{15pt}\usebox\mybox\par}
\end{algorithm}

The $m$ reference points $t_1,\dots,t_m$ are obtained from Algorithm
\ref{alg:sampleindep} for independence testing and Algorithm
\ref{alg:samplehomo} for test of sample proportions. Goodness-of-fit testing is
actually a standard simulation
trick 
(also noted by \cite{gaboardi2016diffchi}):
since goodness-of-fit tests whether the table came from a Multinomial($n_0,
\theta_0$) distribution, for a prespecified $\theta_0$, one simply generates $m$
tables $Q_1,\dots, Q_m$ from this distribution, adds fresh independent
Laplace$(2/\epsilon)$ noise to each table to get $\widetilde{Q}_1,\dots,
\widetilde{Q}_m$, sets $t_i$ to be the value of the test statistic computed from
$\widetilde{Q}_i$, and compares all the $t_i$ to the test statistic $t^*$ on
$\widetilde{T}$.

\begin{restatable}{theorem}{thmgood}
\label{thm:good}\textbf{(Goodness-of-fit).}
Let $T[\cdot]$ be a sample from a Multinomial$(n_0,\theta_0)$
distribution. Let $\widetilde{T}=T+V_\epsilon\kappa \sqrt{n_0}$ where $V_\epsilon$ is
a vector of independent Laplace$(2/\epsilon)$ random variables. If no cells have probability $0$,
then as $n_0\rightarrow\infty$, then
the statistics:
{\small
\begin{align}
\widetilde{\chi^2} &= \sum_j (\widetilde{T}[j]-n_0\theta_0[j])^2/(n_0\theta_0[j])\label{eqn:noisygoodchi}\\
\widetilde{LR} &= 2\sum_j \left[\widetilde{T}[j]\log[\widetilde{T}[j]/(n_0\theta_0[j])] - \widetilde{T}[j]+n_0\theta_0[j]\right]\label{eqn:noisygoodlr} 
\end{align}
}
asymptotically have the same distribution as:
$\sum\limits_j\left( A[j]+\kappa V^*[j] \right)^2/\theta_0[j]$
where $V^*$ has the same distribution as $V_\epsilon$ and $A \sim N(\bm{0}, \diag(\theta_0) - \theta_0\theta_0^t)$.
\end{restatable}
For proof, see \ConfOrTech{the full version of our paper
  \cite{wang2015differentially}}{the online Appendix \ref{app:good}}.

\begin{theorem}
The $p$-value algorithms satisfy $\epsilon$-differential privacy.
\end{theorem}
\begin{proof}
The algorithms only use the differentially private noisy tables, the table sizes (which are assumed to be known according to the definition of differential privacy in Definition \ref{def:dp}), and the density of the noise distribution (which is also public knowledge). At no point do they access the true data or the true noise that was added to it. As such, our algorithms are strictly post-processing algorithms and hence satisfy differential privacy due to the post-processing property of differential privacy \citep{dwork2006calibrating}.
\end{proof}


\section{Experiments} \label{sec:experiments}
\conferenceversion{ 
\begin{figure}
 \centering
 \subfloat{
  \includegraphics[width=0.21\textwidth]{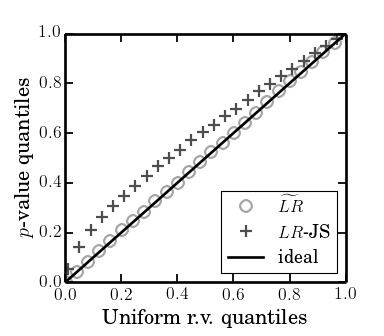}
 }%
 \hskip -5pt
 \subfloat{
  \includegraphics[width=0.21\textwidth]{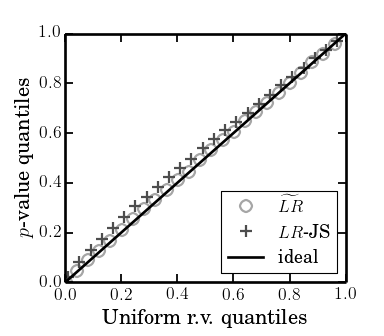}
 }%
 \hskip -5pt
 \subfloat{
  \includegraphics[width=0.21\textwidth]{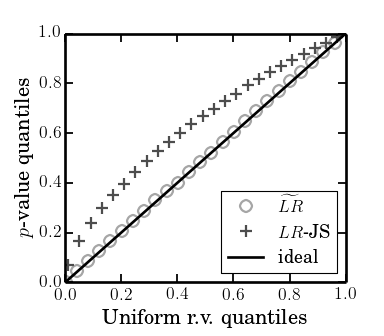}
 }%
 \hskip -5pt
 \subfloat{
  \includegraphics[width=0.21\textwidth]{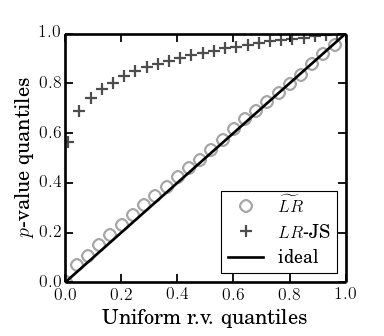}
 }
 \vskip -13pt
 \caption{Q-Q plots against the uniform distribution with $\epsilon=0.2$ for
   test of independence. $n=1000$, $r=c=2$, $P_{row}=P_{col}=[1/2,1/2]$ (left);
   $n=4000$, $r=c=2$, $P_{row}=P_{col}=[1/2,1/2]$ (middle left); $n=4000$,
   $r=c=3$, $P_{row}=P_{col}=[1/3,1/3,1/3]$ (middle right); $n=4000$, $r=c=3$,
   $P_{row}=P_{col}=[0.1,0.1,0.8]$ (right).}
 \label{fig:privatereliability}
\end{figure}
} 

\techreport{ 
} 

\conferenceversion{
\begin{figure}
 \centering
  \subfloat{
    \includegraphics[width=0.32\textwidth]{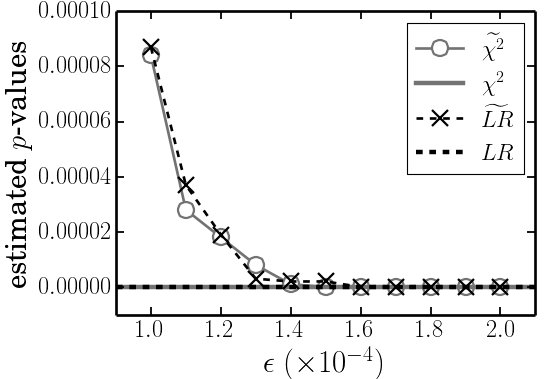}
 }%
 \hskip -5pt
  \subfloat{
    \includegraphics[width=0.32\textwidth]{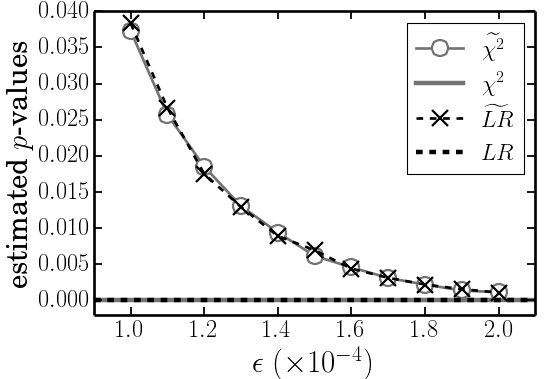}
 }%
 \hskip -5pt
  \subfloat{
    \includegraphics[width=0.32\textwidth]{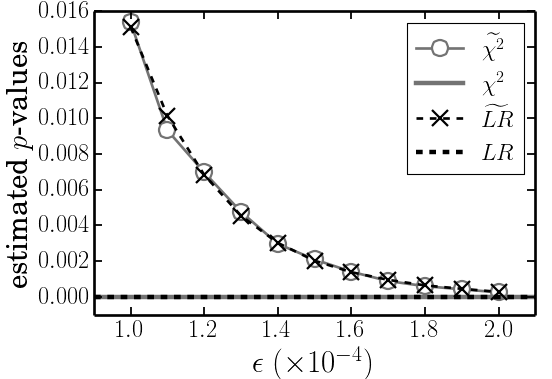}
 }%
 \vskip -14pt
 \caption{Test: independence, Attributes: Passenger Count, Payment Type with
   $n=165,114,361$, $r=4$, $c=3$ (left); Test: sample proportions, Attribute:
   Payment Type (first half year and second half year) with $n_1=85,480,239$,
   $n_2=79,634,122$, $r=2$, $c=3$ (middle); Test: Goodness-of-fit, Attribute:
   Payment Type (second half of year) with $n=79,634,122$, $r=1$, $c=3$,
   \probvec=$[0.00823912,0.5762475,0.41551338]$ estimated from first half of
   year (right).}
 \label{fig:nycpower}
\end{figure}
} 

We evaluate our proposed algorithms on a variety of datasets, both large and small, with various settings of the differential privacy parameter $\epsilon$. 
In particular, we use extremely challenging conditions where the added noise is significant compared to the standard deviation of the data.
Small data sets are commonly seen in the social sciences since the effort of collecting experimental data naturally limits the data size. However, large datasets allow one to use very small $\epsilon$ values (such as $0.0001$) to provide very strong privacy guarantees for individuals. 

The goal of the experiments is to determine at what point the methods
  break down. This will serve to identify the frontier for future research.
Generally, we found that the tests work extremely well on large data even with
large amounts of privacy noise. The tests work well on small datasets (e.g.,
$n\approx 1,800$)  where the non-private $p$-value strongly rejects the null
hypothesis (e.g., $p\leq 0.01$). Beyond that, the \textcolor{black}{agreement with non-private tests} starts to
degrade (a significant improvement over prior work \citep{johnson2013privacy,
 uhler2013privacy,yu2014scalable} in terms of reliability and \textcolor{black}{agreement with non-private tests}).

\textcolor{black}{ We use five real data sets from which we obtain seven
  contingency tables. The first dataset (Czech) was used
  in~\cite{fienberg2010differential}. Its purpose is to study the risk factors
  for coronary thrombosis with data collected from all men employed in a Czech
  car factory at the beginning of the $15$ year follow-up study. Its sample size
  is $1841$. The second dataset (Rochdale) was also used
  in~\cite{fienberg2010differential} and it contains information on $665$
  households in Rochdale, UK, which was used to study the factors that
  influence whether a wife is economically active or not. Its sample size is
  $665$. The third data set was used in~\cite{yang2012differential}. It is a
  synthetic dataset which contains information about home zone, work zone and
  income category of individuals. It was formed using an ad hoc privacy approach
  for data extracted from a $2000$ census database. Its sample size is $2291$.
  The fourth data set was used in~\cite{wright2014intuitive} and contains data
  from the $1972$ National Opinion Research Center General Society Survey about
  white Christians' attitude toward abortion. Its sample size is $1055$. The
  previous datasets are relatively small (and hence very challenging for
  privacy-preserving statistical testing). For a large dataset, which allows us
  to explore very small $\epsilon$ settings, we used the 2014 NYC Taxi data
  (available at
  \url{http://www.nyc.gov/html/tlc/html/about/trip_record_data.shtml}) which
  contains trip records for all NYC yellow taxis in 2014. This dataset has a sample size
  of $165,114,361$. Fig.~\ref{tab:nyctaxi} from online Appendix~\ref{app:datasets}
  summarizes the dataset. Other details for these datasets can
  be found from online Appendix~\ref{app:datasets}.} We evaluate the reliability of
our methods in Section \ref{subsec:exprely}. We evaluate the quality of the
$p$-values on the real datasets in Section
\ref{subsec:expreal}. 
Our experiments use
10,000 reference points to compute $p$-values and $p$-value results in Section
\ref{subsec:expreal} are averaged over 100
runs (of privacy-preserving table perturbations). 

\subsection{Reliability of the Asymptotic Distributions} \label{subsec:exprely}

A $p$-value is a strong statistical statement: under the null hypothesis,
$P(\text{$p$-value}\leq q)=q$. In other words, under the null hypothesis,
$p$-values must be uniformly distributed. This criterion allows us to evaluate
the reliability of tests: pick a null distribution, sample data from it, compute
$p$-values from each sampled dataset, and plot the quantiles of the resulting
$p$-values against the uniform distribution. The result is a Q-Q plot and a
perfect match will be a diagonal line. We compare the reliability of our
$p$-value computations for the $\chi^2$ and likelihood ratio statistics (denoted
by $\widetilde{\chi^2}$ and $\widetilde{LR}$) against the reliability of the
earlier method proposed by \cite{johnson2013privacy}, which simply ran noisy
tables through off-the-shelf software (we denote their results by $\chi^2$-JS
and LR-JS). \textcolor{black}{In each Q-Q plot, we only show every $400^{\text{th}}$ points from the total $10,000$ points for each test in case we cannot distinguish between the markers due to stacking.} 
\conferenceversion{We present results for $\widetilde{LR}$ vs. LR-JS as the
graphs are virtually identical for $\widetilde{\chi^2}$ vs. $\chi^2$-JS.
Reliability plots for sample proportions and goodness-of-fit tests also look
very similar and are omitted due to lack of space. The omitted plots can all
be found in the full version \citep{wang2015differentially}.}

\techreport{
\begin{figure}
 \centering
 \subfloat{
  \includegraphics[width=0.32\textwidth]{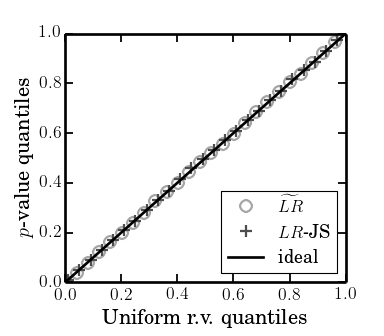}
 }%
 \hskip -6pt
 \subfloat{
  \includegraphics[width=0.32\textwidth]{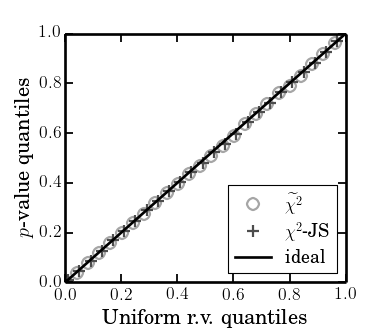}
 }%
 \hskip -6pt
 \subfloat{
  \includegraphics[width=0.32\textwidth]{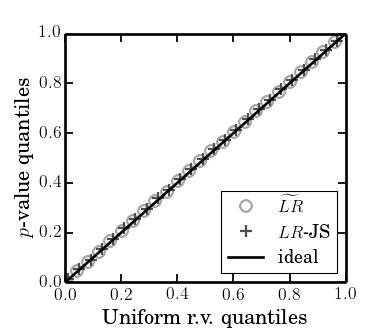}
 }%
 \vskip -16pt
 \caption{Q-Q plots against the uniform distribution with $\epsilon=\infty$ (no
   noise). Test of independence (left two), $n_0=1000$, $r=c=2$,
   $P_{row}=P_{col}=[1/2,1/2]$. Test of sample proportions (right), $n_1=1200$,
   $n_2=2800$, table dimension$=2$, \probvec{}$=[1/2,1/2]$. Similar results are
   obtained for other parameter settings. Both methods are expected to perform
   well for this sanity check.}
 \label{fig:nonprivatereliability}
\end{figure}

\begin{figure}
 \centering
 \subfloat{
  \includegraphics[width=0.32\textwidth]{independence-pv-n1000r2c2skew0eps2teststat0.png}
 }%
 \hskip -6pt
 \subfloat{
  \includegraphics[width=0.32\textwidth]{independence-pv-n4000r3c3skew0eps2teststat0.png}
 }%
 \hskip -6pt
 \subfloat{
  \includegraphics[width=0.32\textwidth]{independence-pv-n4000r3c3skew1eps2teststat0.png}
 }%
 \vskip -13pt
 \subfloat{
  \includegraphics[width=0.32\textwidth]{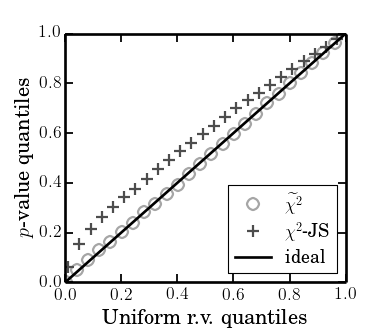}
 }%
 \hskip -6pt
 \subfloat{
  \includegraphics[width=0.32\textwidth]{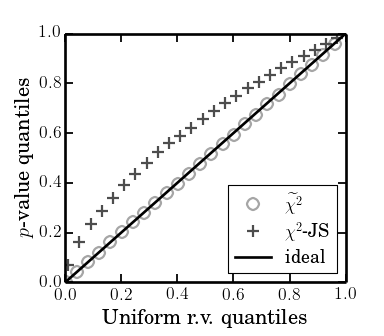}
 }%
 \hskip -6pt
 \subfloat{
  \includegraphics[width=0.32\textwidth]{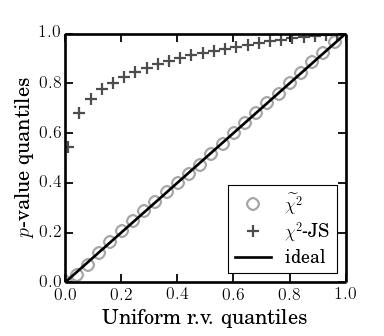}
 }%
 \vskip -16pt
 \caption{Q-Q plots against the uniform distribution with $\epsilon=0.2$ for
   test of independence. $n_0=1000$, $r=c=2$, $P_{row}=P_{col}=[1/2,1/2]$ (left);
   $n_0=4000$, $r=c=3$, $P_{row}=P_{col}=[1/3,1/3,1/3]$ (middle); $n_0=4000$,
   $r=c=3$, $P_{row}=P_{col}=[0.1,0.1,0.8]$ (right).}
 \label{fig:privatereliability}
\end{figure}
}

\techreport{ 
When no noise is added to the tables, both methods match the ideal diagonal line
and so are reliable, as expected and shown in
Fig.~\ref{fig:nonprivatereliability} (left two) for independence testing and
Fig.~\ref{fig:nonprivatereliability} (right) for test of sample proportions.
}


Fig. \ref{fig:privatereliability} shows the Q-Q plots under a variety of settings, such as sample size $n_0$, number of rows $r$, number of columns $c$ and the type of null distribution: the null distribution probability of table entry $T[i,j]$ is set to $P_{row}[i]P_{col}[j]$. As can be seen from Fig. \ref{fig:privatereliability}, our methods  remain reliable throughout these settings while the competitors  returned unreliable $p$-values. The upward bend of the reliability curve for  LR-JS (and $\chi^2$-JS) indicates strong bias towards producing small $p$-values and hence would lead to many false discoveries if applied in practice.

\techreport{ 

\begin{figure}
 \centering
 \subfloat{
  \includegraphics[width=0.32\textwidth]{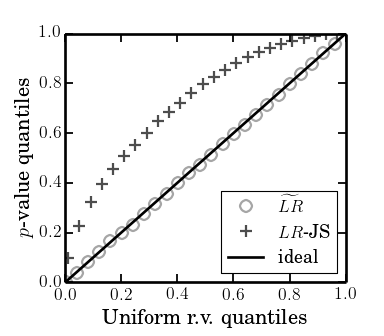}
 }%
 \hskip -6pt
 \subfloat{
  \includegraphics[width=0.32\textwidth]{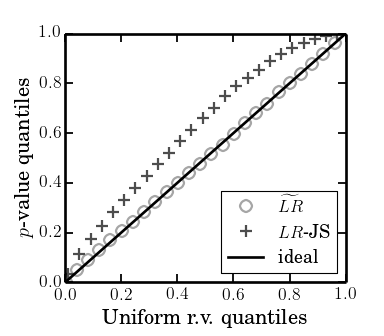}
 }%
 \hskip -6pt
 \subfloat{
  \includegraphics[width=0.32\textwidth]{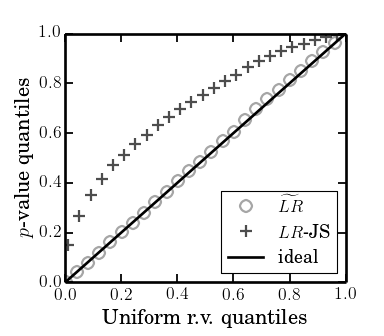}
 }%
 \vskip -13pt
 \subfloat{
  \includegraphics[width=0.32\textwidth]{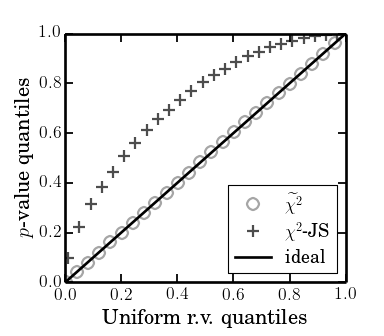}
 }%
 \hskip -6pt
 \subfloat{
  \includegraphics[width=0.32\textwidth]{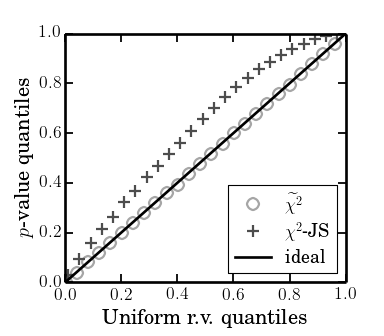}
 }%
 \hskip -6pt
 \subfloat{
  \includegraphics[width=0.32\textwidth]{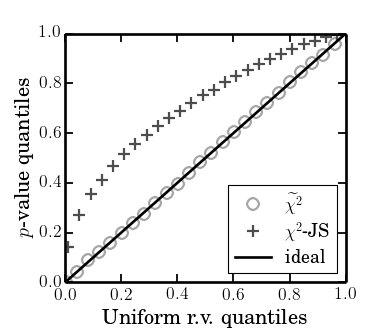}
 }%
 \vskip -16pt
 \caption{Q-Q plots against the uniform distribution with $\epsilon=0.2$ for
   test of sample proportions. $n_1=400$, $n_2=600$, table dimension$=2$,
   \probvec{}$=[1/2,1/2]$ (left); $n_1=1200$, $n_2=2800$, table dimension$=2$,
   \probvec{}$=[1/2,1/2]$ (middle);
   $n_1=1200$, $n_2=2800$, table dimension$=3$, \probvec{}$=[0.1,0.1,0.8]$
   (right).}
 \label{fig:noisyhomoreliability}
\end{figure}

\begin{figure}
 \centering
 \subfloat{
  \includegraphics[width=0.32\textwidth]{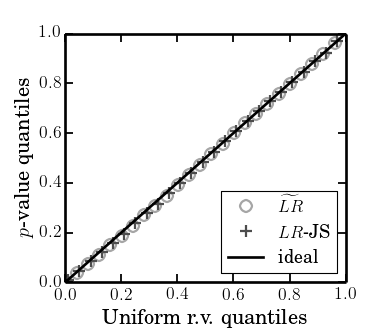}
 }%
 \hskip -6pt
 \subfloat{
  \includegraphics[width=0.32\textwidth]{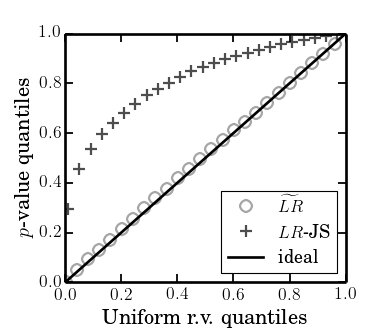}
 }%
 \hskip -6pt
 \subfloat{
  \includegraphics[width=0.32\textwidth]{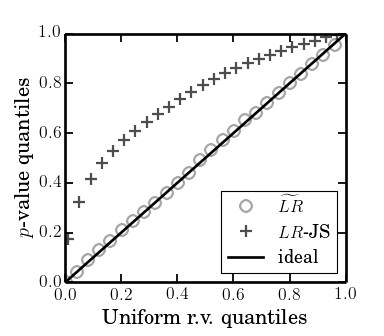}
 }%
 \vskip -13pt
 \subfloat{
  \includegraphics[width=0.32\textwidth]{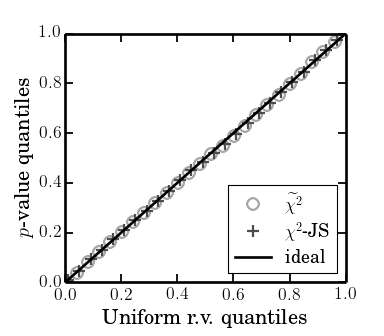}
 }%
 \hskip -6pt
 \subfloat{
  \includegraphics[width=0.32\textwidth]{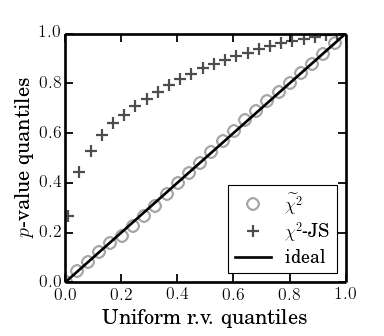}
 }%
 \hskip -6pt
 \subfloat{
  \includegraphics[width=0.32\textwidth]{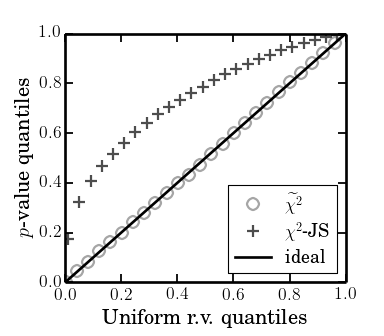}
 }%
 \vskip -16pt
 \caption{Q-Q plots against the uniform distribution for goodness of fit test.
   $n_0=1000$, $r=1$, $c=4$, \probvec{}$=[1/4,1/4,1/4,1/4]$, $\epsilon=\infty$
   (left); $n_0=500$, $r=1$, $c=4$, \probvec{}$=[1/4,1/4,1/4,1/4]$, $\epsilon=0.2$
   (middle);
   $n_0=1000$, $r=1$, $c=4$, \probvec{}$=[0.1,0.2,0.3,0.4]$, $\epsilon=0.2$
   (right).}
 \label{fig:goodreliability}
\end{figure}

Reliability plots for the test of sample proportions can be found in
Fig.~\ref{fig:noisyhomoreliability} and goodness-of-fit in
Fig.~\ref{fig:goodreliability}. Recall that test of sample proportions tests
whether two tables come from the same multinomial distribution. In the figure,
we report the settings of the two table sizes $n_1$ and $n_2$, the number of
cells in each table, and the null distribution Multinomial probability vector
\probvec{} used to generate tables for the reliability plot. The various
settings for the goodness-of-fit reliability plots are presented in
Fig.~\ref{fig:goodreliability}.

\textcolor{black}{Those reliability results show that under privacy settings the naive method proposed by \cite{johnson2013privacy} is not reliable at all, and suggest that our $p$-value computations are more reliable and thus should be used instead.}
} 

\subsection{P-value comparison on Real Data}\label{subsec:expreal}

\techreport{
\begin{figure}
 \centering
  \subfloat{
    \includegraphics[width=0.32\textwidth]{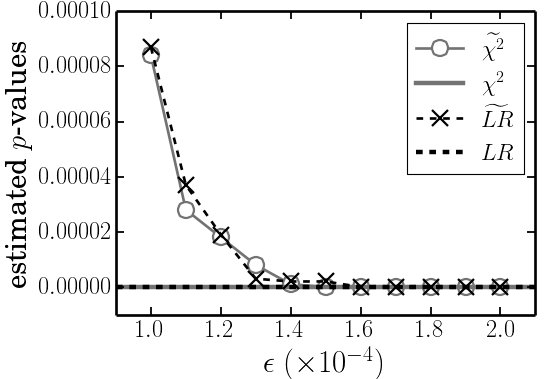}
 }%
 \hskip -5pt
  \subfloat{
    \includegraphics[width=0.32\textwidth]{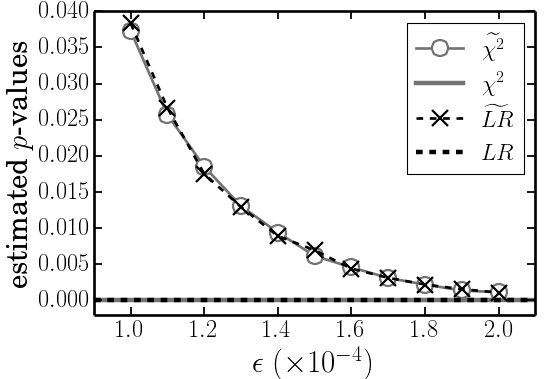}
 }%
 \hskip -5pt
  \subfloat{
    \includegraphics[width=0.32\textwidth]{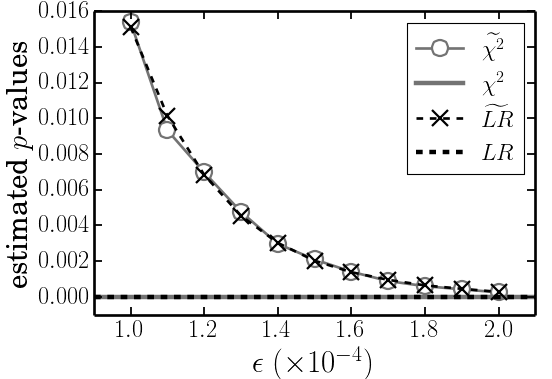}
 }%
 \vskip -14pt
 \caption{Test: independence, Attributes: Passenger Count, Payment Type with
   $n_0=165,114,361$, $r=4$, $c=3$ (left); Test: sample proportions, Attribute:
   Payment Type (first half year and second half year) with $n_1=85,480,239$,
   $n_2=79,634,122$, $r=2$, $c=3$ (middle); Test: Goodness-of-fit, Attribute:
   Payment Type (second half of year) with $n_0=79,634,122$, $r=1$, $c=3$,
   \probvec=$[0.00823912,0.5762475,0.41551338]$ estimated from first half of
   year (right).}
 \label{fig:nycpower}
\end{figure}
} 

Now we evaluate the $p$-values generated by our methods using real datasets
and compare them to non-private $p$-values. \textcolor{black}{All private $p$-values are averages over 100 repetitions.} Please note that the experimental settings
are challenging as the standard deviation of the added noise $(\sqrt{8}/\epsilon)$ is substantial compared to the standard deviation of the data itself ($O(\sqrt{n})$).
For example, the taxi data has sample size $n_0=165,114,361$ and we used $\epsilon=0.0001$. Here $\sqrt{n_0}=12,849.7$ and noise std$=28,284.3$ so a loss in \textcolor{black}{agreement with non-private tests} is expected.
Nevertheless, we generally find that when the null hypothesis is strongly rejected (non-private $\chi^2$ or LR $p$-value $\leq 0.01$), our private tests $\widetilde{\chi^2}$ and $\widetilde{LR}$ also reject the null hypothesis at level $0.01$. The output perturbation methods of \cite{uhler2013privacy,yu2014scalable} required too much noise and are omitted to avoid skewing the graphs. \neweat{ (we evaluate them in the permutation testbed, where their noise requirement is smaller).}
Fig. \ref{fig:nycpower} shows our results for various tests on large NYC taxi data. The $p$-values show very good performance\textcolor{black}{, as the null hypothesis is rejected (as with the original data) even when the privacy noise is extremely large ($\epsilon=0.0001$)}.


%

Next we move to extremely challenging cases with small sample sizes but high
relative noise. We arrange the figures so that from left to right there is a
decrease in sample size and a reduction in statistical significance of
\emph{non-private} analysis. Fig.~\ref{fig:realtable1} and
Fig.~\ref{fig:homopower2} show results for independence testing and test of
sample proportions, respectively\conferenceversion{( additional experiments can
  be found in \cite{wang2015differentially})}. The results show good \textcolor{black}{agreement with the non-private tests} when
the null hypothesis is strongly rejected (non-private $p\leq 0.01$) with sample
sizes of $1,800$ or more. Agreement decreases as sample size decreases and
non-private $p$-value increases. Of particular interest is the Rochdale data
(Fig.~\ref{fig:realtable1}, right). It is a small dataset with a very high
$p$-value and small test statistic that is dominated by noise. Since the noise
obscures any statistical signal, the $p$-value behaves more like a uniform
random variable (hence we add 80\% error bars on the plot). This is expected,
and note that the null hypothesis would be erroneously rejected only in rare
circumstances (which is the expected behavior of $p$-values).

\begin{figure}
 \centering
 \subfloat{
  \includegraphics[width=0.32\textwidth]{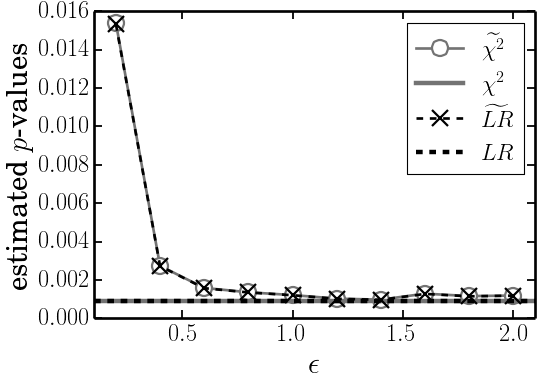}
   \label{fig:realtable1a}
 }%
 \hskip -5pt
 \subfloat{
  \includegraphics[width=0.32\textwidth]{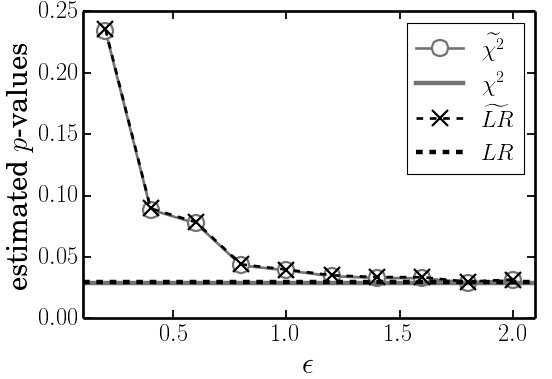}
    \label{fig:realtable1b}
 }%
 \hskip -5pt
 \subfloat{
  \includegraphics[width=0.32\textwidth]{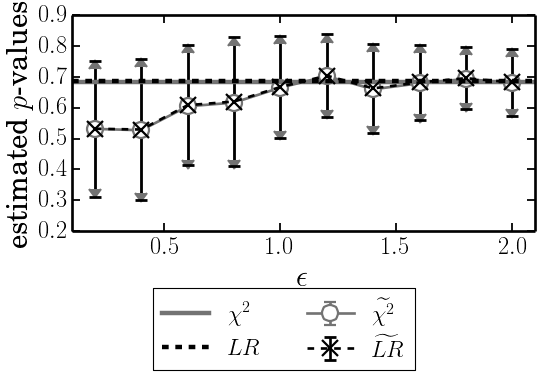}
    \label{fig:realtable1c}
 }%
 \vskip -13pt
 \caption{Independence testing on $2 \times 2$ tables. Tables used:
   Czech \texttt{AD} with $n_0=1841$ (left), Czech \texttt{BC} with $n_0=1841$
   (middle), Rochdale \texttt{AB} with $n_0=665$ (right).}
 \label{fig:realtable1}
\end{figure}

\techreport{ 
Fig.~\ref{fig:realtable2} shows the results for independence tests with Census data. Again, as we move from left to right we observe smaller data sizes and less (non-private) evidence against the null hypothesis seem to reduce \textcolor{black}{agreement with non-private tests}. Fig.~\ref{fig:homopower1} shows the results for test of sample
proportions with Czech car worker and Rochdale data. Again, as we move from left
to right we observe smaller data sizes and less (non-private) evidence against
the null hypothesis \textcolor{black}{seem} to reduce \textcolor{black}{agreement with non-private tests}. The extreme
case is again the Rochdale data where the small non-private test statistic gets
dominated by noise. The resulting $p$-value is closer to being uniformly
distributed as noise gets larger.

\begin{figure}
 \centering
 \subfloat{
   \includegraphics[width=0.32\textwidth]{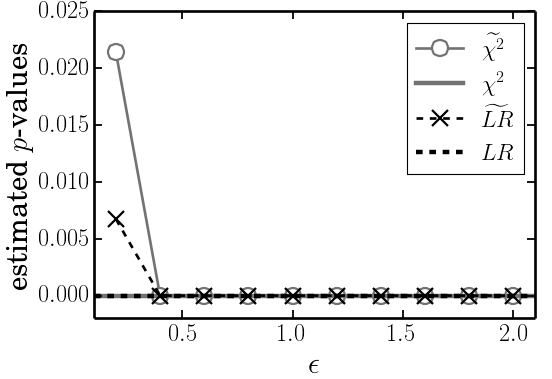}
 }%
 \hskip -5pt
  \subfloat{
    \includegraphics[width=0.32\textwidth]{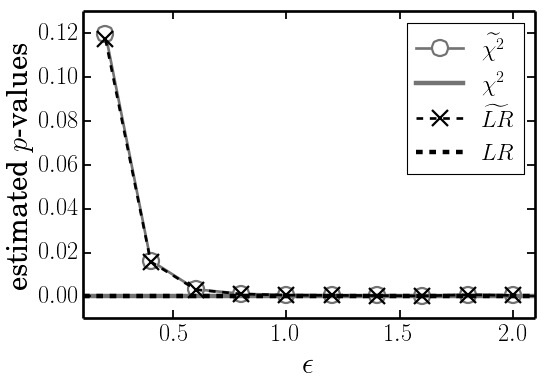}
 }%
 \hskip -5pt
  \subfloat{
    \includegraphics[width=0.32\textwidth]{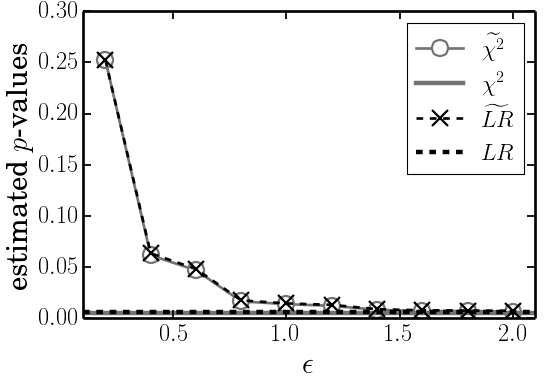}
 }%
 \vskip -13pt
 \caption{Independence tests. Tables used: Home Zone and Income Category with
   $n_0=2291$, $r=4$, $c=16$ (left), Religion and Attitudes with $n_0=1055$, $r=c=3$
   (middle), Religion and Education with $n_0=1055$, $r=c=3$ (right).}
 \label{fig:realtable2}
\end{figure}

\begin{figure}
 \centering
 \subfloat{
  \includegraphics[width=0.32\textwidth]{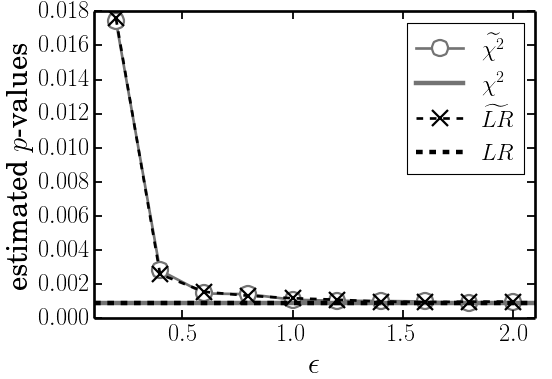}
 }%
 \hskip -5pt
 \subfloat{
  \includegraphics[width=0.32\textwidth]{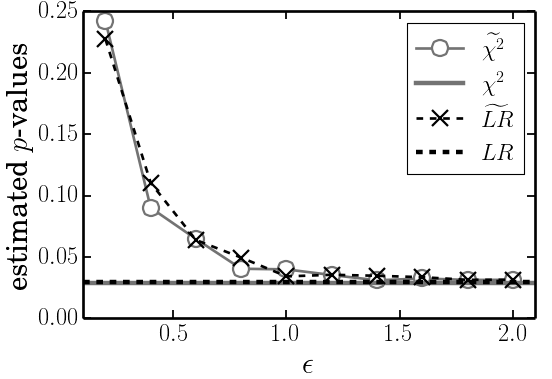}
 }%
 \hskip -5pt
 \subfloat{
  \includegraphics[width=0.32\textwidth]{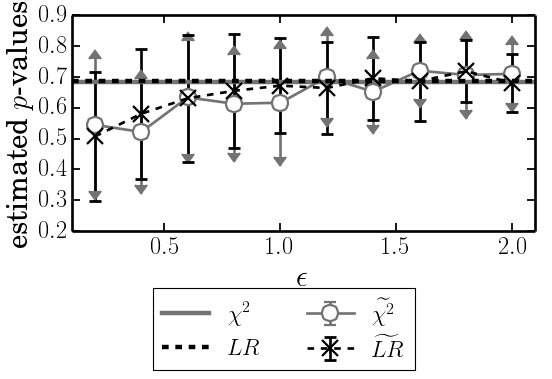}
 }%
 \vskip -14pt
 \caption{Test of sample proportions (Table dim $=2$). Tables used:
   Czech \texttt{A} with $n_1=1054$, $n_2=$ $787$ (left), Czech \texttt{B} with
   $n_1=1581$, $n_2=260$ (middle), Rochdale \texttt{A} with $n_1=586$, $n_2=79$
   (right).}
 \label{fig:homopower1}
\end{figure}


} 
\begin{figure}
 \centering
 \subfloat{
  \includegraphics[width=0.32\textwidth]{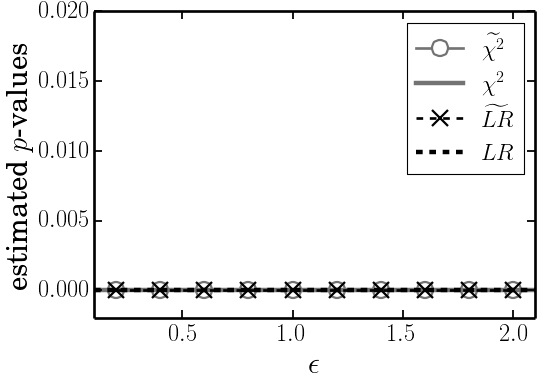}
 }%
 \hskip -5pt
 \subfloat{
  \includegraphics[width=0.32\textwidth]{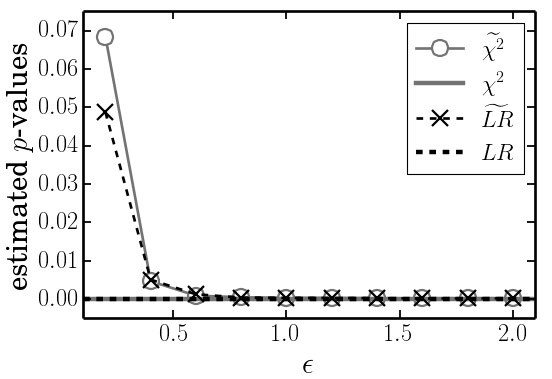}
 }%
 \hskip -5pt
 \subfloat{
  \includegraphics[width=0.32\textwidth]{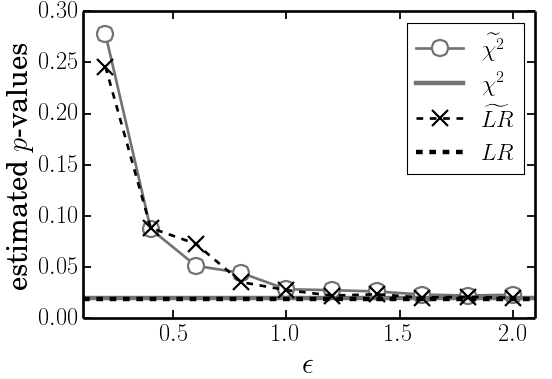}
 }%
 \vskip -13pt
 \caption{Test of sample proportions. Attributes used: Home Zone with
   $n_1=1286$, $n_2=1005$, Table dim $=4$ (left), Attitudes with $n_1=453$,
   $n_2=602$, Table dim $=3$ (middle), Education with $n_1=453$, $n_2=602$,
   Table dim $=3$ (right).}
 \label{fig:homopower2}
\end{figure}

\techreport{ 
\begin{figure}
 \centering
 \subfloat{
  \includegraphics[width=0.32\textwidth]{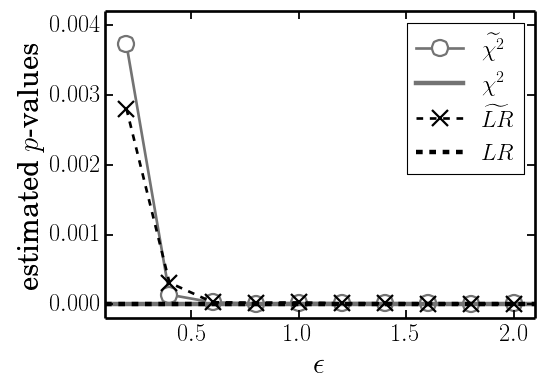}
 }%
 \hskip -5pt
 \subfloat{
  \includegraphics[width=0.32\textwidth]{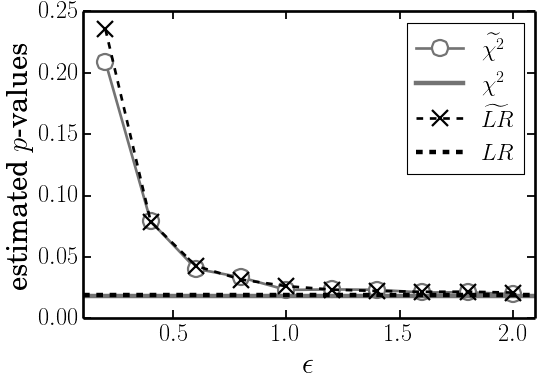}
 }%
 \hskip -5pt
 \subfloat{
  \includegraphics[width=0.32\textwidth]{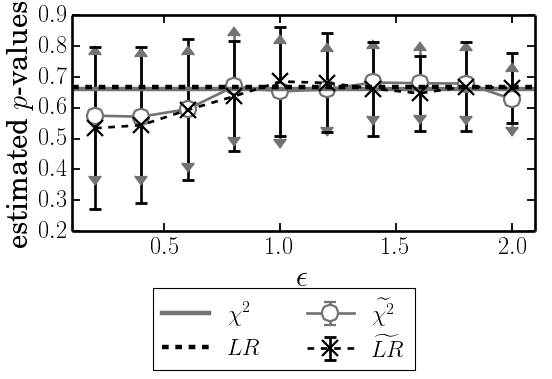}
 }%
 \vskip -13pt
 \caption{Goodness of fit on $1 \times 2$ tables. Tables
   used: Czech \texttt{A} with $n_0=787$, \probvec{}=$[0.4886148,0.5113852]$
   (left), Czech \texttt{B} with $n_0=260$, \probvec=$[0.58760278, 0.41239722]$
   (middle), Rochdale \texttt{A} with $n_0=79$, \probvec=$[0.77986348,
   0.24013652]$ (right).}
 \label{fig:goodpower1}
\end{figure}

\begin{figure}
 \centering
 \subfloat{
  \includegraphics[width=0.32\textwidth]{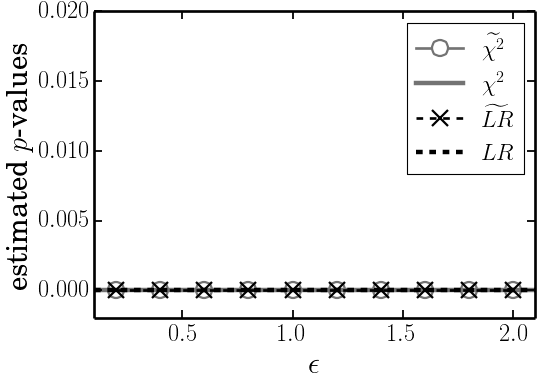}
 }%
 \hskip -5pt
 \subfloat{
  \includegraphics[width=0.32\textwidth]{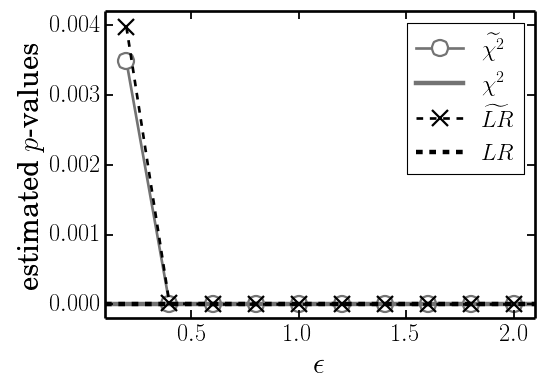}
 }%
 \hskip -5pt
 \subfloat{
  \includegraphics[width=0.32\textwidth]{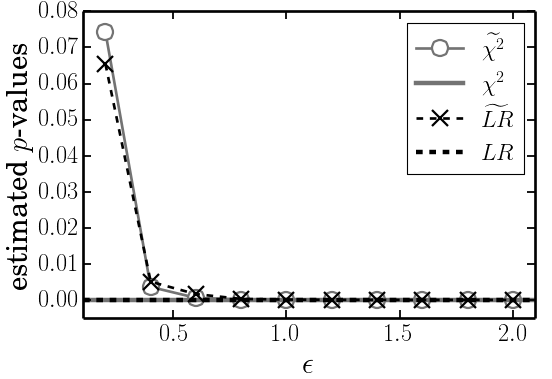}
 }%
 \vskip -13pt
 \caption{Goodness of fit tests. Attributes used: Home Zone with $n_0=1005$, $r=1$, $c=4$,
   \probvec{}=$[0.50155521, 0.05987558, 0.21772939,0.22083981]$ (left),
   Attitudes with $n_0=602$, $r=1$, $c=3$, \probvec{}=$[0.37748344, 0.21633554,
   0.40618102]$ (middle), Education with $n_0=602$, $r=1$, $c=3$,
   \probvec{}=$[0.14569536, 0.53421634, 0.3200883]$ (right).}
 \label{fig:goodpower2}
\end{figure}

For goodness of fit tests, we used extremely small sample sizes (on the order of
a few hundred) with large standard deviation of Laplace noise relative to the
standard deviation of the data. The results are shown in Figs.
\ref{fig:goodpower1} and \ref{fig:goodpower2}. The private $p$-values are still
in good quality but their quality degrades as the sample size is further
diminished. Again, the Rochdale data with only 79 data points has a very small
non-private $\chi^2$ and LR value and so is completely dominated by noise,
leading to large variance but no false conclusions except in rare cases (as is
allowed by the definition of $p$-value).

\textcolor{black}{We have several observations from the above results. For very large datasets, our tests show very good agreement with the non-private tests even with very rigorous differential privacy guarantee. When it comes to the cases with small sample sizes but high relative noise, our tests perform very well when there is strong signal against the null hypothesis. Agreement with non-private tests appears to decrease as sample size decreases and non-private $p$-value increases. Even for very small datasets where the statistical signal is dominated by noise, our tests only lead to false conclusions in rare cases.}
} 

\section{Conclusions}\label{sec:conc}
In this paper, we revisited the topic of $\epsilon$-differentially private hypothesis testing. We provided $p$-value algorithms that, for the first time, allow reliable private hypothesis tests for data sizes often used in social sciences, as well as reliable tests with very strong privacy protections (i.e. small $\epsilon$ values) for large data sizes. \textcolor{black}{The advantages of our algorithms over previous approaches have been verified through the extensive experiments.} We believe that new test statistics tailored for the privacy domain may yield even further improvement, but those test statistics are open problems. \neweat{ We also proposed a permutation-based testbed to help estimate the quality of new test statistics on real data before fully working out their mathematical details.}

\bibliographystyle{rss}

\begin{thebibliography}{23}
\expandafter\ifx\csname natexlab\endcsname\relax\def\natexlab#1{#1}\fi
\expandafter\ifx\csname url\endcsname\relax
  \def\url#1{\texttt{#1}}\fi
\expandafter\ifx\csname urlprefix\endcsname\relax\def\urlprefix{URL: }\fi

\bibitem[{Chaudhuri and Hsu(2012)}]{chaudhuri2012convergence}
Chaudhuri, K. and Hsu, D. (2012) Convergence rates for differentially private
  statistical estimation.
\newblock In \textit{ICML}.

\bibitem[{Dwork(2006)}]{dworkdiffp2006}
Dwork, C. (2006) Differential privacy.
\newblock In \textit{ICALP}.

\bibitem[{Dwork and Lei(2009)}]{dwork2009differential}
Dwork, C. and Lei, J. (2009) Differential privacy and robust statistics.
\newblock In \textit{STOC}.

\bibitem[{Dwork et~al.(2006)Dwork, McSherry, Nissim and
  Smith}]{dwork2006calibrating}
Dwork, C., McSherry, F., Nissim, K. and Smith, A. (2006) Calibrating noise to
  sensitivity in private data analysis.
\newblock In \textit{TCC}.

\bibitem[{Dwork et~al.(2015)Dwork, Su and Zhang}]{dwork2015private}
Dwork, C., Su, W. and Zhang, L. (2015) Private false discovery rate control.
\newblock \textit{arXiv:1511.03803}.

\bibitem[{Ferguson(1996)}]{largesamplebook}
Ferguson, T.~S. (1996) \textit{A Course in Large Sample Theory}.
\newblock Chapman \& Hall.

\bibitem[{Fienberg et~al.(2010)Fienberg, Rinaldo and
  Yang}]{fienberg2010differential}
Fienberg, S.~E., Rinaldo, A. and Yang, X. (2010) Differential privacy and the
  risk-utility tradeoff for multi-dimensional contingency tables.
\newblock In \textit{PSD}.

\bibitem[{Gaboardi et~al.(2016)Gaboardi, Lim, Rogers and
  Vadhan}]{gaboardi2016diffchi}
Gaboardi, M., Lim, H., Rogers, R. and Vadhan, S. (2016) Differentially private
  chi-squared hypothesis testing: Goodness of fit and independence testing.
\newblock In \textit{ICML}.

\bibitem[{Good(2004)}]{permutationtest}
Good, P. (2004) \textit{Permutation, Parametric, and Bootstrap Tests of
  Hypotheses, 3rd ed}.
\newblock Springer.

\bibitem[{Homer et~al.(2008)Homer, Szelinger, Redman, Duggan, Tembe, Muehling,
  John V.~Pearson, Nelson and Craig}]{homer:genome}
Homer, N., Szelinger, S., Redman, M., Duggan, D., Tembe, W., Muehling, J., John
  V.~Pearson, D. A.~S., Nelson, S.~F. and Craig, D.~W. (2008) Resolving
  individuals contributing trace amounts of dna to highly complex mixtures
  using high-density snp genotyping microarrays.
\newblock \textit{Plos Genetics}, \textbf{4}.

\bibitem[{Johnson and Shmatikov(2013)}]{johnson2013privacy}
Johnson, A. and Shmatikov, V. (2013) Privacy-preserving data exploration in
  genome-wide association studies.
\newblock In \textit{KDD}.

\bibitem[{Kifer and Machanavajjhala(2014)}]{pufferfish}
Kifer, D. and Machanavajjhala, A. (2014) Pufferfish: A framework for
  mathematical privacy definitions.
\newblock \textit{ACM Trans. Database Syst.}, \textbf{39}, 3:1--3:36.

\bibitem[{Lei(2011)}]{lei2011differentially}
Lei, J. (2011) Differentially private m-estimators.
\newblock In \textit{NIPS}.

\bibitem[{Machanavajjhala and Kifer(2015)}]{designprivacy}
Machanavajjhala, A. and Kifer, D. (2015) Designing statistical privacy for your
  data.
\newblock \textit{Commun. ACM}, \textbf{58}, 58--67.

\bibitem[{Nissim et~al.(2007)Nissim, Raskhodnikova and Smith}]{smoothsens}
Nissim, K., Raskhodnikova, S. and Smith, A. (2007) Smooth sensitivity and
  sampling in private data analysis.
\newblock In \textit{STOC}.

\bibitem[{Smith(2011)}]{smith2011privacy}
Smith, A. (2011) Privacy-preserving statistical estimation with optimal
  convergence rates.
\newblock In \textit{STOC}.

\bibitem[{Uhler et~al.(2013)Uhler, Slavkovic and Fienberg}]{uhler2013privacy}
Uhler, C., Slavkovic, A. and Fienberg, S.~E. (2013) Privacy-preserving data
  sharing for genome-wide association studies.
\newblock \textit{Journal of Privacy and Confidentiality}, \textbf{5}.

\bibitem[{Vu and Slavkovic(2009)}]{vu2009differential}
Vu, D. and Slavkovic, A. (2009) Differential privacy for clinical trial data:
  Preliminary evaluations.
\newblock In \textit{ICDM Workshops}.

\bibitem[{Wang et~al.(2015)Wang, Lee and Kifer}]{wang2015differentially}
Wang, Y., Lee, J. and Kifer, D. (2015) Differentially private hypothesis
  testing, revisited.
\newblock \textit{arXiv:1511.03376}.

\bibitem[{Wasserman and Zhou(2010)}]{wassermanmini}
Wasserman, L. and Zhou, S. (2010) A statistical framework for differential
  privacy.
\newblock \textit{Journal of the American Statistical Association},
  \textbf{105}, 375--389.

\bibitem[{Wright and Smucker(2014)}]{wright2014intuitive}
Wright, S.~E. and Smucker, B.~J. (2014) An intuitive formulation and solution
  of the exact cell-bounding problem for contingency tables of conditional
  frequencies.
\newblock \textit{Journal of Privacy and Confidentiality}, \textbf{5}, 4.

\bibitem[{Yang et~al.(2012)Yang, Fienberg and Rinaldo}]{yang2012differential}
Yang, X., Fienberg, S.~E. and Rinaldo, A. (2012) Differential privacy for
  protecting multi-dimensional contingency table data: Extensions and
  applications.
\newblock \textit{Journal of Privacy and Confidentiality}, \textbf{4}, 5.

\bibitem[{Yu et~al.(2014)Yu, Fienberg, Slavkovi{\'c} and
  Uhler}]{yu2014scalable}
Yu, F., Fienberg, S.~E., Slavkovi{\'c}, A.~B. and Uhler, C. (2014) Scalable
  privacy-preserving data sharing methodology for genome-wide association
  studies.
\newblock \textit{Journal of biomedical informatics}, \textbf{50}, 133--141.

\end{thebibliography}

\techreport{
\newpage
\onecolumn
\appendix
\section{Proof of Theorem \lowercase{\ref{thm:asympguarantee}}}\label{app:asympguarantee}
\thmasympguarantee*
\begin{proof}
For vectors, we will use the notation $\vec{z}\prec\vec{t}$ to mean each component of $\vec{z}$ is less than or equal to the corresponding component of $\vec{t}$. Hence a cumulative distribution function $F_A(\vec{t})$ for a vector-valued random variable $A$ represents $P(A\prec\vec{t})$.

Since $X_0$ and $Y$ may be discrete random variables while the Gaussian is continuous, let $\mu$ be a suitable measure (e.g., mixture of Lebesgue and counting measure) so that we can write $f_Y$ to be the Radon-Nikodym derivative of the $Y$ with respect to $\mu$,  $f_0$ to be the Radon-Nikodym derivative of $X_0$ with respect to $\mu$ and $f_X$ be the Radon-Nikodym derivative of $N(0,\sigma^2)$ with respect to $\mu$. Then, for any $\vec{t}$:
\begin{itemize}
\item The density of $X_{n_0}/\sqrt{n_0}$ is $f^*(\vec{x})\equiv\sqrt{n_0}f_0(\sqrt{n_0}\vec{x})$
\item The density of $Z_{n_0}/\sqrt{n_0}$ is then $g_0(\vec{z})\equiv \int_{\vec{x}}f^*(\vec{x})\sqrt{n_0}f_Y(\sqrt{n_0}(\vec{z}-\vec{x}))~d\mu(\vec{x})=\int_{\vec{x}}\sqrt{n_0}f_0(\sqrt{n_0}\vec{x})\sqrt{n_0}f_Y(\sqrt{n_0}(\vec{z}-\vec{x}))~d\mu(\vec{x})$
\end{itemize}

We claim that for all $\vec{t}$, $G_Z(\vec{t}) = \int_{\vec{z}\prec\vec{t}}\int_{\vec{x}} \phi(\vec{x})\sqrt{n_0}f_Y(\sqrt{n_0}(\vec{z}-\vec{x})) ~d\mu(\vec{z})~d\mu(\vec{x})$, which is the CDF of the convolution of the Gaussian with $Y/\sqrt{n_0}$. This follows immediately from Slutsky's Theorem because the $X_{n}$ are independent from $Y$. Then,

\begin{align*}
\lefteqn{|G_0(\vec{t}) - G_Z(\vec{t})|}\\ &=\left|\int_{\vec{z}\prec\vec{t}}\int_{\vec{x}} \Big(
\sqrt{n_0}f_0(\sqrt{n_0}\vec{x})\sqrt{n_0}f_Y(\sqrt{n_0}(\vec{z}-\vec{x}))
-\phi(\vec{x})\sqrt{n_0}f_Y(\sqrt{n_0}(\vec{z}-\vec{x}))\Big)d\mu(\vec{x})~d\mu(\vec{z})\right|\\
&=\left|\int_{\vec{z}\prec\vec{t}}\int_{\vec{x}} \Big(
\sqrt{n_0}f_0(\sqrt{n_0}\vec{x})
-\phi(\vec{x})\Big)\sqrt{n_0}f_Y(\sqrt{n_0}(\vec{z}-\vec{x}))d\mu(\vec{x})~d\mu(\vec{z})\right|\\
&=\left|\int_{\vec{z}\prec\vec{t}}\int_{\vec{x}} \Big(
\sqrt{n_0}f_0(\sqrt{n_0}(\vec{z}-\vec{x}))
-\phi(\vec{z}-\vec{x})\Big)\sqrt{n_0}f_Y(\sqrt{n_0}\vec{x})d\mu(\vec{x})~d\mu(\vec{z})\right|\\
&=\left|\int_{\vec{x}} \Big(
F_0(\vec{t}-\vec{x})
-\Phi(\vec{t}-\vec{x})\Big)\sqrt{n_0}f_Y(\sqrt{n_0}\vec{x})~d\mu(\vec{x})\right|\\
&\leq \int_{\vec{x}} \left|
F_0(\vec{t}-\vec{x})
-\Phi(\vec{t}-\vec{x})\right|\sqrt{n_0}f_Y(\sqrt{n_0}\vec{x})~d\mu(\vec{x})\\
&\leq \left(\sup_{\vec{x}} \left|
F_0(\vec{x})
-\Phi(\vec{x})\right|\right)\int_{\vec{x}}\sqrt{n_0}f_Y(\sqrt{n_0}\vec{x})~d\mu(\vec{x})\\ 
&=\sup_{\vec{x}} \left|
F_0(\vec{x})
-\Phi(\vec{x})\right|
\end{align*}
\end{proof}
\section{Proof of Theorem \lowercase{\ref{thm:chilr}}}\label{app:chilr}
\thmchilr*

Without loss of generality, assume all of the noisy tables have been stuffed into one vector $\widetilde{T}$ and similarly for the expected counts $E$.

Note that asymptotically, the correction of terms where $\widetilde{T}[i]<0$ will not be needed since $\widetilde{T}[i]/n$ converges in probability to the true parameter $\theta[i]$.
\begin{proof}
We will use the Taylor series expansion of $\log(1+x)$ around $x=0$:
\begin{align*}
\log(1+x) &= x - \int_0^1\int_0^1 u\frac{1}{(1+uvx)^2}x^2~dv~du
\end{align*}
\begin{align*}
& 2\sum_i \left(\widetilde{T}[i]\log \frac{\widetilde{T}[i]}{E[i]} - \widetilde{T}[i] + E[i]\right)\\
=&  2\sum_i \widetilde{T}[i]\left(\frac{\widetilde{T}[i]}{E[i]}-1\right) - 2\sum_i\widetilde{T}[i]\int_0^1\int_0^1 u\left(\frac{\widetilde{T}[i]}{E[i]}-1\right)^2\frac{1}{(1+uv\left(\frac{\widetilde{T}[i]}{E[i]}-1\right) )^2}~dv~du\\
\quad& - 2\sum_i E[i]\left(\frac{\widetilde{T}[i]}{E[i]}-1\right)\\
=&2\sum_i (\widetilde{T}[i]-E[i])\left(\frac{\widetilde{T}[i]}{E[i]}-1\right) - 2\sum_i\widetilde{T}[i]\int_0^1\int_0^1 u\left(\frac{\widetilde{T}[i]}{E[i]}-1\right)^2\frac{1}{(1+uv\left(\frac{\widetilde{T}[i]}{E[i]}-1\right) )^2}~dv~du\\
=&2\sum_i \frac{(\widetilde{T}[i]-E[i])^2}{E[i]} - 2\sum_i\frac{\widetilde{T}[i]}{E[i]}\int_0^1\int_0^1 u\frac{(\widetilde{T}[i]-E[i])^2}{E[i]}\frac{1}{(1+uv\left(\frac{\widetilde{T}[i]}{E[i]}-1\right) )^2}~dv~du\\
\end{align*}

Now, both $\widetilde{T}/n$ and $E/n$ converge in probability to the true (nonzero) null distribution and so $\widetilde{T}/E$ converges to 1 in probability. Thus, an application of Slutsky's theorem \citep{largesamplebook} allows us to conclude that the term containing the integral converges in distribution to the same limit as $2\sum_i\int_0^1\int_0^1 u\frac{(\widetilde{T}[i]-E[i])^2}{E[i]}~dv~du= \sum_i \frac{(\widetilde{T}[i]-E[i])^2}{E[i]}$.

Thus $2\sum_i \left(\widetilde{T}[i]\log \frac{\widetilde{T}[i]}{E[i]} - \widetilde{T}[i] + E[i]\right)$ and $\frac{(\widetilde{T}[i]-E[i])^2}{E[i]}$ converge in distribution to the same limit.
\end{proof}

\section{Proof of Theorem \lowercase{\ref{thm:indep}}}\label{app:indep}
\thmindep*

We first need the following Lemma~\ref{lem:clt}.
\begin{lemma} \label{lem:clt} Let $T$ be a contingency table sampled from a
  Multinomial$(n,\theta)$ distribution with no entries having $0$ probability. Let $V_\epsilon$ be a table (with same dimensions
  as $T$) of independent Laplace$(2/\epsilon)$ random variables. Let $\widetilde{T}=T+V_{\epsilon}\kappa \sqrt{n}$. Then as $n
  \rightarrow \infty$,
  $\frac{\widetilde{T}-n\theta}{\sqrt{n}}$ converges in law to the distribution
  of the random variable $A+\kappa V^*$, where $V^*$ has the same distribution as $V_\epsilon$ and $vec(A) \sim N(0, \diag(vec(\theta)) - vec(\theta)vec(\theta)^t)$ 
\end{lemma}
\begin{proof}
Since $T$ and $V_\epsilon$ are independent, the result follows from the Central limit theorem and a variation of Slutsky's theorem \citep{largesamplebook}.
\end{proof}

The proof of Theorem~\ref{thm:indep} is provided below.

\begin{proof} \label{pf:asymptotic}
When convenient, we will treat $\widetilde{T}$, $T$ and $\theta$ as either vectors (with one index) or 2-d arrays  with two indices (e.g.,  $\theta[i,j]$). The conversion is simple: $\theta[(i-1)*c + j]=\theta[i,j]$.

 We will consider the noisy likelihood ratio statistic as it is easier to work with (we apply Theorem \ref{thm:chilr} and note that in the theorem, $E[i,j]=\widetilde{T}[i,\bullet]\widetilde{T}[\bullet,j]/\widetilde{T}[\bullet,\bullet]$ and so $\sum_{ij} E[i,j]-\sum_{ij}\widetilde{T}[i,j]=0$):
\begin{align*}
 \widetilde{LR} = & 2\left(\sum_i\sum_j \widetilde{T}[i,j]\log\left(\frac{\widetilde{T}[i,j]}{E[i,j]}\right)\right)\\
 =& 2\sum\limits_{i=1}^r\sum\limits_{j=1}^c \widetilde{T}[i,j]\log\left(\frac{\widetilde{T}[i,j]}{\sum\limits_{i^*=1}^r\sum\limits_{j^*=1}^c \widetilde{T}[i^*,j^*]}\right)
  - 2\sum\limits_{i=1}^r\sum\limits_{j=1}^c \widetilde{T}[i,j]\log\left(\frac{\sum\limits_{i^*=1}^r \widetilde{T}[i^*,j]}{\sum\limits_{i^*=1}^r\sum\limits_{j^*=1}^c \widetilde{T}[i^*,j^*]}\right)\\
  &- 2\sum\limits_{i=1}^r\sum\limits_{j=1}^c \widetilde{T}[i,j]\log\left(\frac{\sum\limits_{j^*=1}^c \widetilde{T}[i,j^*]}{\sum\limits_{i^*=1}^r\sum\limits_{j^*=1}^c \widetilde{T}[i^*,j^*]}\right)\\
 =&2\sum\limits_{i=1}^r\sum\limits_{j=1}^c \widetilde{T}[i,j]\log\left(\widetilde{T}[i,j]\right)
  +2\left(\sum\limits_{i=1}^r\sum\limits_{j=1}^c \widetilde{T}[i,j]\right)\log\left(\sum\limits_{i^*=1}^r\sum\limits_{j^*=1}^c \widetilde{T}[i^*,j^*]\right)\\
 &-2\sum\limits_{j=1}^c \left(\sum\limits_{i=1}^r \widetilde{T}[i,j]\right)\log\left(\sum\limits_{i^*=1}^r \widetilde{T}[i^*,j]\right)
  -2\sum\limits_{i=1}^r\left(\sum\limits_{j=1}^c \widetilde{T}[i,j]\right)\log\left(\sum\limits_{j^*=1}^c \widetilde{T}[i,j^*]\right)
 \end{align*}

 We will
 \begin{itemize}
 \item use the second order taylor expansion to expand these quantities around
   $n\theta_0[i,j]$, $n\theta_0[i,\cdot]$ $n\theta_0[\cdot, j]$ and
   $n\theta_0[\cdot,\cdot]$. That is, $f(x)= f(x_0) + (x-x_0)^t\nabla f(x_0) +
   (x-x_0)^t \int_0^1\int_0^1 \nabla^2 vf(x_0 + uv(x-x_0))~du~dv (x-x_0)$.
 \item use the fact that $\theta_0[i,j]=\theta_0[i,\cdot]\theta_0[\cdot,j]$.
 \item use convergence in probability of $\widetilde{T}/n$ to $\theta_0$  to deduce that
   $n/[n\theta_0[i,j] + uv(\widetilde{T}[i,j]-n\theta_0[i,j])]\rightarrow 1/\theta_0[i,j] $ in probability.
 \end{itemize}
 
 {\small
 \begin{align*}
 &\widetilde{LR}\\
 =& 2\sum\limits_{ij}n\theta_0[i,j]\log(n\theta_0[i,j]) - 2\sum\limits_i n\theta_0[i,\cdot]\log(n\theta_0[i,\cdot])
 - 2\sum\limits_j n\theta_0[\cdot,j]\log(n\theta_0[\cdot,j]) \\
 & + 2n \theta_0[\cdot,\cdot]\log(n\theta_0[\cdot,\cdot])
  +2\sum\limits_{ij}(\widetilde{T}[i,j] - n\theta_0[i,j])(1+\log(n\theta_0[i,j]))\\
 & - 2\sum\limits_i (\widetilde{T}[i,\cdot]-n\theta_0[i,\cdot])(1+\log(n\theta_0[i,\cdot])) 
  - 2\sum\limits_j (\widetilde{T}[\cdot,j]-n\theta_0[\cdot,j])(1+\log(n\theta_0[\cdot,j])) \\
 & + 2(\widetilde{T}[\cdot,\cdot] - n \theta_0[\cdot,\cdot])(1+\log(n\theta_0[\cdot,\cdot]))\\
 & +2\sum\limits_{ij}\int_0^1\int_0^1 v\frac{(\widetilde{T}[i,j]-n\theta_0[i,j])^2}{n\theta_0[i,j] + uv(\widetilde{T}[i,j]-n\theta_0[i,j])}~du~dv
 -2\sum\limits_{i}\int_0^1\int_0^1 v\frac{(\widetilde{T}[i,\cdot]-n\theta_0[i,\cdot])^2}{n\theta_0[i,\cdot] + uv(\widetilde{T}[i,\cdot]-n\theta_0[i,\cdot])}~du~dv\\
 & -2\sum\limits_{j}\int_0^1\int_0^1 v\frac{(\widetilde{T}[\cdot,j]-n\theta_0[\cdot,j])^2}{n\theta_0[\cdot,j] + uv(\widetilde{T}[\cdot,j]-n\theta_0[\cdot,j])}~du~dv
 + 2\int_0^1\int_0^1 v\frac{(\widetilde{T}[\cdot,\cdot]-n\theta_0[\cdot,\cdot])^2}{n\theta_0[\cdot,\cdot] + uv(\widetilde{T}[\cdot,\cdot]-n\theta_0[\cdot,\cdot])}~du~dv\\
 =& 2\sum\limits_{ij}\int_0^1\int_0^1 v\frac{(\widetilde{T}[i,j]-n\theta_0[i,j])^2}{n\theta_0[i,j] + uv(\widetilde{T}[i,j]-n\theta_0[i,j])}~du~dv
 -2\sum\limits_{i}\int_0^1\int_0^1 v\frac{(\widetilde{T}[i,\cdot]-n\theta_0[i,\cdot])^2}{n\theta_0[i,\cdot] + uv(\widetilde{T}[i,\cdot]-n\theta_0[i,\cdot])}~du~dv\\
 & -2\sum\limits_{j}\int_0^1\int_0^1 v\frac{(\widetilde{T}[\cdot,j]-n\theta_0[\cdot,j])^2}{n\theta_0[\cdot,j] + uv(\widetilde{T}[\cdot,j]-n\theta_0[\cdot,j])}~du~dv
 + 2\int_0^1\int_0^1 v\frac{(\widetilde{T}[\cdot,\cdot]-n\theta_0[\cdot,\cdot])^2}{n\theta_0[\cdot,\cdot] + uv(\widetilde{T}[\cdot,\cdot]-n\theta_0[\cdot,\cdot])}~du~dv\\
 \sim& \sum\limits_{ij}\left(\frac{\widetilde{T}[i,j]-n\theta_0[i,j]}{\sqrt{n}}\right)^2\frac{1}{\theta_0[i,j]}
  -\sum\limits_{i}\left(\frac{\widetilde{T}[i,\cdot]-n\theta_0[i,\cdot]}{\sqrt{n}}\right)^2\frac{1}{\theta_0[i,\cdot]}\\
&  -\sum\limits_{j}\left(\frac{\widetilde{T}[\cdot,j]-n\theta_0[\cdot,j]}{\sqrt{n}}\right)^2\frac{1}{\theta_0[\cdot,j]} 
  + \left(\frac{\widetilde{T}[\cdot,\cdot]-n\theta_0[\cdot,\cdot]}{\sqrt{n}}\right)^2\frac{1}{\theta_0[\cdot,\cdot]}\\
\sim& \sum\limits_{ij}\frac{\left(A[i,j]+\kappa V^*[i,j]\right)^2}{\theta_0[i,j]}
  -\sum\limits_{i}\frac{\left(A[i,\cdot]+\kappa V^*[i,\cdot]\right)^2}{\theta_0[i,\cdot]}
 -\sum\limits_{j}\frac{\left(A[\cdot,j]+\kappa V^*[\cdot,j]\right)^2}{\theta_0[\cdot,j]} 
  + \frac{\left(A[\cdot,\cdot]+\kappa V^*[\cdot,\cdot]\right)^2}{\theta_0[\cdot,\cdot]}
 \end{align*}
}

 Where the last two lines follow from: (a) noting that under the null hypothesis
 $\theta_0[i,j]=\theta_0[i,\cdot]\theta_0[\cdot,j]$, (b) convergence in probability
 (e.g., $\frac{\widetilde{T}[1,1]}{\sum_{ij} \widetilde{T}[i,j]}\rightarrow \theta_0[1,1]$), (c)
 Lemma~\ref{lem:clt} and (d) Slutsky's theorem \citep{largesamplebook}. 

 The theorem follows for the likelihood ratio statistic. The asymptotic equivalence follows from convergence in probability. Together with
 Theorem~\ref{thm:chilr}, the theorem also follows for the chi-squared statistic.
\end{proof}

\section{Proof of Theorem \lowercase{\ref{thm:homo}}}\label{app:homo}
\thmhomo*

\begin{proof}
Note that $E_1[i] = \frac{n_1(\widetilde{T}[j]+\widetilde{S}[j])}{n_1+n_2}$ and $E_2[i]=\frac{n_2(\widetilde{T}[j]+\widetilde{S}[j])}{n_1+n_2}$ and $\sum_i E_1[i]+\sum_j E_2[j] = \sum_i \widetilde{T}[i] + \sum_j \widetilde{S}[j]$ (which we use when applying Theorem \ref{thm:chilr}).

According to Definition~\ref{def:homogeneity}, the chi-squared statistic based
on $\widetilde{T}$ and $\widetilde{S}$ is:
$$\widetilde{\chi^2}=\sum_j\frac{\left(
    \widetilde{T}[j]-\dfrac{n_1(\widetilde{T}[j]+\widetilde{S}[j])}{n_1+n_2}
  \right)^2}{\dfrac{n_1(\widetilde{T}[j]+\widetilde{S}[j])}{n_1+n_2}} +
\sum_j\frac{\left(
    \widetilde{S}[j]-\dfrac{n_2(\widetilde{T}[j]+\widetilde{S}[j])}{n_1+n_2}
  \right)^2}{\dfrac{n_2(\widetilde{T}[j]+\widetilde{S}[j])}{n_1+n_2}}$$


\begin{align*}
&\widetilde{\chi^2} \\
=& 
\sum\limits_j\left[ \dfrac{n_1+n_2}{n_1\zdj}\left( \widetilde{T}[j]-\dfrac{n_1\zdj}{n_1+n_2} \right)^2 + \dfrac{n_1+n_2}{n_2\zdj}\left( \widetilde{S}[j]-\dfrac{n_2\zdj}{n_1+n_2} \right)^2 \right]
\\
=& \sum\limits_j\left[ \dfrac{n_1+n_2}{n_1\zdj}\left( \dfrac{n_2\widetilde{T}[j]}{n_1+n_2}-\dfrac{n_1\widetilde{S}[j]}{n_1+n_2} \right)^2 + \dfrac{n_1+n_2}{n_2\zdj}\left( \dfrac{n_1\widetilde{S}[j]}{n_1+n_2}-\dfrac{n_2\widetilde{T}[j]}{n_1+n_2} \right)^2 \right]
\\
=& \sum\limits_j\left( \dfrac{1}{n_1\myn\zdj} + \dfrac{1}{n_2\myn\zdj} \right)\left( n_2\widetilde{T}[j] - n_1\widetilde{S}[j] \right)^2
\\
=& \sum\limits_j\dfrac{n_1n_2}{\zdj}\left( \dfrac{\widetilde{T}[j]}{n_1}-\thetaj+\thetaj-\dfrac{\widetilde{S}[j]}{n_2} \right)^2
\\
=& \sum\limits_j\dfrac{n_1n_2}{\myn}\dfrac{1}{\dfrac{\zdj}{\myn}}\left[ \dfrac{1}{\sqrt{n_1}}\whichzj{1}{\widetilde{T}} - \dfrac{1}{\sqrt{n_2}}\whichzj{2}{\widetilde{S}} \right]^2
\\
=& \sum\limits_j\dfrac{1}{\dfrac{\zdj}{\myn}}\left[ \whichn{2}\whichzj{1}{\widetilde{T}}- \whichn{1}\whichzj{2}{\widetilde{S}} \right]^2
\\
\sim & \sum\limits_j\frac{1}{\thetaj}\left[ \whichn{2}\left(A_1[j]+V^*_1[j]\right) - \whichn{1}\left(A_2[j]+V^*_2[j]\right)\right]^2
\end{align*}
where the last line follows from a) convergence in probability $\frac{(\widetilde{T}[j]+\widetilde{S}[j])}{n_1+n_2}\rightarrow \text{\probvec{}}[j]$, (b) Lemma~\ref{lem:clt} in Appendix~\ref{app:indep}), and (c) Slutsky's
theorem \citep{largesamplebook}.

The theorem follows for the chi-squared statistic. Together with
Theorem~\ref{thm:chilr}, the theorem also follows for the likelihood ratio statistic. The asymptotic equivalence also follows from convergence in probability.
\end{proof}

\section{Proof of Theorem \lowercase{\ref{thm:good}}}\label{app:good}
\thmgood*

\begin{proof}
According to Definition~\ref{def:goodness}, the chi-squared statistic based on $\widetilde{T}$ is:
\begin{align*}
\widetilde{\chi^2}&=\sum_j\frac{(\widetilde{T}[j]-n\text{\probvec{}}[j])^2}{n\text{\probvec{}}[j]}
=\sum_j\frac{\left( \dfrac{\widetilde{T}[j]-n\text{\probvec{}}[j]}{\sqrt{n}} \right)^2}{\text{\probvec{}}[j]}
\sim\sum_j\frac{\left( A[j]+\kappa V^*[j] \right)^2}{\text{\probvec{}}[j]}
\end{align*}
because of Lemma~\ref{lem:clt} in
Appendix~\ref{app:indep}.

The theorem follows for the chi-squared statistic. Together with
Theorem~\ref{thm:chilr}, the theorem also follows for the likelihood ratio statistic.
\end{proof}

\newpage
\section{Datasets Used in Experiments}\label{app:datasets}
\textcolor{black}{ We provide details for the $5$ datasets used in the
  experiments in this section.}

\textcolor{black}{ The first dataset (Czech) was used
  in~\cite{fienberg2010differential} and it was collected from all men employed
  in a Czech car factory at the beginning of a $15$ year follow-up study. It was
  used to study the risk factors for coronary thrombosis. Its sample size is
  $1841$. There are $6$ binary attributes in the dataset. A: smoking, B:
  strenuous mental work, C: strenuous physical work, D: systolic blood pressure,
  E: ratio of $\beta$ and $\alpha$ lipoproteins, F: family anamnesis of coronary
  heart disease.}

\textcolor{black}{ The second dataset (Rochdale) was also used
  in~\cite{fienberg2010differential} and it contains information from $665$
  households in Rochdale, UK. It was used to study the factors which influence
  whether a wife is economically active or not. Its sample size is $665$. There
  are $8$ binary attributes in the dataset. A: wife employed? (yes, if wife is
  economically active), B: wife's age $>38$? C: husband employed? D: child?
  (yes, if there is a child of age $<4$ in the household), E: wife's education
  is O-level+? F: husband's education is O-level+? G: Asian origin? H: household
  working? (yes, if any other member than wife or husband of the household is
  working).}

\textcolor{black}{ The third dataset was used in~\cite{yang2012differential}. It
  is a synthetic dataset which contains information about home zone, work zone
  and income category of individuals. It was formed using an ad hoc privacy
  approach for data extracted from a $2000$ census database. Its sample size is
  $2291$. There are $3$ categorical variables, with $4$ zones of origin (Home
  Zone), $4$ zones of destination (Work Zone) and $16$ income categories (Income
  Category).}

\textcolor{black}{ The fourth dataset was used in~\cite{wright2014intuitive} and
  contains data from the $1972$ National Opinion Research Center General Society
  Survey about white Christians' attitude toward abortion. Its sample size is
  $1055$. There are $3$ categorical variables with $3$ religions (Religion), $3$
  groups of education years (Education) and $3$ attitudes (Attitudes).}

\textcolor{black}{ The fifth dataset contains all NYC yellow taxi trip data in
  $2014$ (available at
  \url{http://www.nyc.gov/html/tlc/html/about/trip_record_data.shtml}). This
  dataset has a sample size of $165,114,361$. There are $2$ categorical
  variables: passenger count (Passenger Count) and payment type (Payment Type).
  Fig.~\ref{tab:nyctaxi} summarizes the dataset.}

  \begin{figure}
    \centering
    \begin{tabular}{|c|ccc|}
      \hline
      \multicolumn{1}{|c|}{\multirow{2}{*}{Passenger Count}} & \multicolumn{3}{c|}{Payment Type} \\ \cline{2-4}
      \multicolumn{1}{|c|}{} & CRD & CSH & Others \\ \hline
      1 & 68685857 & 46625277 & 980220 \\
      2 & 12711902 & 10180961 & 166088 \\
      3-4 & 5232235 & 5043192 & 82001 \\
      Others & 8941327 & 6318250 & 147051 \\ \hline
    \end{tabular}
    \caption{2014 NYC yellow taxi trip data (There are $2$ categorical
      variables: passenger count and payment type). Its sample size is
      $165,114,361$.}
    \label{tab:nyctaxi}
  \end{figure}

\clearpage
\section{Permutation Testbed}\label{sec:permutation}
The likelihood ratio and $\chi^2$ statistics are general-purpose statistical tools that can be adapted to a variety of tests. However, it is likely that some unknown privacy-specific test statistics could outperform them. Finding such test statistics is an open problem and, as we have seen, approximating their null distributions will generally not be easy. It would be very helpful to be able to estimate how well a new test statistic could perform \emph{on real data} before figuring out all of these mathematical details. This is the purpose of our permutation-based testbed.

First, we need to choose an application that is rich enough to exhibit variations between various test statistics. Naively, one could select the goodness-of-fit test since it is possible to exactly sample from its null distribution (see Section \ref{sec:asymptotics}). However, goodness-of-fit is too simple: rejecting the null hypothesis is the same as rejecting 1 possible distribution for the data. On the other hand, independence testing is a much richer scenario. The null hypothesis (independence of rows and columns) consists of infinitely many possible distributions (those where rows and columns are independent) and rejecting the null hypothesis means rejecting all of these distributions.

We argue that the ideal testbed is differentially-private independence testing when row and column sums are known (in other words, it protects information about individuals modulo what can be learned from the marginals). In practice, releasing row and column sums could cause a breach of privacy. Hence, this is only a method for exploring how test statistics would behave under statistical disclosure control (and for most applications it is not to be used for actually releasing analytical results). As a testbed, it has several appealing properties (which we explain in this section):
\begin{list}{$\bullet$}{
\setlength{\itemsep}{0pt}
\setlength{\topsep}{3pt}
\setlength{\parsep}{3pt}
\setlength{\partopsep}{0pt}
\setlength{\leftmargin}{1em}}
\item The standard permutation test of independence provides a null sampling distribution that works with any test statistic, hence there is no need to derive asymptotic approximations for the testbed.
\item Real-data experimental results on input-perturbation methods would carry over straightforwardly to the unrestricted case (i.e. unknown row and column sums).
\item Real-data experimental results on output-perturbation methods would be a lower bound on the unrestricted case, hence allowing experimenters to rule out statistics that do not perform well in the testbed.
\end{list}

First, we explain how to do differentially-private hypothesis testing when row and column sums are known in Section \ref{subsec:testbedexplain}. Then we discuss how experimental results would carry over to the unrestricted case in Section \ref{subsec:testbedcarry}. We illustrate its use in Section \ref{subsec:testbedusage}. We present experiments in Section \ref{subsec:expperm} to compare likelihood ratio and $\chi^2$  with other statistics to confirm our intuition that there exist other test statistics that are better suited for privacy. Earlier work by \cite{uhler2013privacy} and  \cite{yu2014scalable} studied differential privacy when only exact column sums (but not row sums) are known -- this scenario would not be suitable for a testbed as the exact null distribution cannot be sampled from. 

\subsection{Private Independence Testing with Known Marginals}\label{subsec:testbedexplain}
Consider a set of records $D=\set{x_1,\dots, x_n}$ and suppose the records have two distinguished categorical attributes, which we call $R$ and $C$. We can construct a table $T[\cdot,\cdot]$ where $T[i,j]$ is the number of records with $R=i$ and $C=j$. Although $T$ is not public knowledge, suppose its row and column sums are known (i.e. for any $i,j$, $T[\bullet, j]$ and $T[i,\bullet]$ are public). We are interested in publishing the rest of the information in $T$ in a private manner that does not leak any more information beyond what the row and column sums already revealed. The privacy definition that allows us to do this was proposed by \cite{pufferfish} -- it ends up being a variant of differential privacy based on a concept called \emph{marginal neighbors}.

\begin{definition}\emph{(Marginal Neighbors \citep{pufferfish}).}\label{def:neighbors}
Two datasets $D_1, D_2$ are marginal neighbors if $D_2$ can be obtained from $D_1$ by swapping some of the attributes between $2$ records from $D_1$. Two tables $T_1[\cdot,\cdot], T_2[\cdot,\cdot]$ are marginal neighbors if they are tabulated from datasets that are marginal neighbors. 
\end{definition}

\begin{definition}\emph{($\epsilon$-MN-Differential Privacy \citep{pufferfish}).}\label{def:neighdp}
$\randalg$ satisfies $\epsilon$-mn-differential privacy if for all $T$ and $T^\prime$ that are marginal neighbors (and have the same row and column sums as the true data) and for  all $V \subseteq \range(\randalg)$,
$$P(\randalg(T) \in V) \leq e^\epsilon  P(\randalg(T^\prime) \in V)$$
\end{definition}

Achieving $\epsilon$-mn-differential privacy is straightforward:
\begin{lemma}[\cite{pufferfish}]\label{lem:mnlaplace}
Given a vector-valued function $h$, $\epsilon$-mn-differential privacy can be achieved by adding independent Laplace$(s_h/\epsilon)$ noise to each component of $h(T)$, where $s_h\geq \max ||h(T_1)-h(T_2)||_1$ (and the max is over all marginal neighbors $T_1,T_2$ having the same row/column sums as the true data). In particular, if $h(T)$ just outputs the table $T$, then $s_h=4$.
\end{lemma}

Given a test statistic $h$ and a real table $T$, the testbed works as follows:
\begin{list}{$\bullet$}{
\setlength{\itemsep}{0pt}
\setlength{\topsep}{3pt}
\setlength{\parsep}{3pt}
\setlength{\partopsep}{0pt}
\setlength{\leftmargin}{1em}}
\item[1.] If one is interested in input perturbation, compute the private statistic value $t^* = h(T + V_\epsilon)$, where $V_\epsilon$ is a random table that ensures $\epsilon$-mn-differential privacy (e.g., a table of independent Laplace$(4/\epsilon)$ random variables). For output perturbation, set $t^*=h(T) + V_\epsilon$, where $V_\epsilon$ is a noisy random variable (such as Laplace$(s_h/\epsilon)$) that ensures $\epsilon$-mn-differential privacy.
\item[2.] Generate multiple pseudo-tables $T_1,\dots, T_m$. Conceptually we do this by creating two urns: urn $U_r$ containing $T[1,\bullet]$ balls labeled "1", $T[2,\bullet]$ balls labeled "2", etc.; and urn $U_c$ containing $T[\bullet,1]$ balls labeled "1", etc.
Generate $n$ samples $r_1,\dots, r_n$ from $U_r$ without replacement, generate $c_1,\dots,c_n$ from $U_c$ without replacement. Set the data $D_i=\set{(r_1,c_1), (r_2,c_2), \dots}$ and tabulate $T_i$ from $D_i$. These are samples from the null hypothesis of independence and are probabilistically equivalent to randomly permuting the $R$ attribute in the original data $D$ \citep{permutationtest}. 
\item[3.] Compute the test statistic value $t_i$ for each of the $T_i$ using the exact same procedure as in Step 1, but using fresh noise and operating on $T_i$ instead of $T$.
\item[4.] Set the $p$-value to be $|\set{t_i:~ t_i\geq t^*}|/m$.
\end{list}

\begin{theorem}The testbed satisfies $\epsilon$-mn-differential privacy.\end{theorem}
\begin{proof}
Step 1 satisfies $\epsilon$-mn-differential privacy by construction and the rest of the steps just use public data (hence they are just post-processing steps). By the post-processing property \citep{dwork2006calibrating}, the whole procedure satisfies $\epsilon$-mn-differential privacy.
\end{proof}

\subsection{Translating Experimental Results}\label{subsec:testbedcarry}

Let us compare input perturbation noise for $\epsilon$-differential privacy and for $\epsilon$-mn-differential privacy. We need Laplace$(2/\epsilon)$ noise for the former and Laplace$(4/\epsilon)$ noise for the latter (i.e. the tables are twice as noisy). Thus $\epsilon$-differential privacy and $2\epsilon$-mn-differential privacy use the same amount of noise and therefore would generate the same values for the test statistic $t^*$ (in the case of input perturbation), so the $p$-value one gets under this testbed using  $2\epsilon$-mn-differential privacy  should correspond to the $p$-value an experimenter would have gotten under $\epsilon$-differential privacy (had statistical details, such as approximating the null distribution with \emph{unknown} row/columns sums, been worked out in advance).

For output perturbation, we add Laplace$(s_h/\epsilon)$ noise for $\epsilon$-mn differential privacy (where $s_h$ is defined in Lemma \ref{lem:mnlaplace}) and Laplace$(\mathcal{S}(h)/\epsilon)$ noise for $\epsilon$-differential privacy (where $\mathcal{S}(h)$ is defined in Definition \ref{def:sensitivity}). Thus the noise added under this testbed using $\frac{s_h\epsilon}{\mathcal{S}(h)}$-mn differential privacy is equal to the noise added under $\epsilon$-differential privacy (and hence sets up the correspondence between the resulting $p$-values). What if one hasn't yet fully worked out the sensitivity $\mathcal{S}(h)$ under $\epsilon$-differential privacy? In this case, the statements are slightly less precise, but by comparing the following: 
\begin{align*}
s_h &= \max_{T_1,T_2 \text{ are marginal neighbors} } ||h(T_1) - h(T_2)||_1 \\
s^* &= \max_{T_1, T_2: \text{ underlying datasets differ on 2 records}} ||h(T_1) - h(T_2)||_1\\
\mathcal{S}(h) &= \max_{T_1, T_2:\text{ underlying datasets differ on 1 record}} ||h(T_1) - h(T_2)||_1
\end{align*}
it follows directly from the definitions (and the triangle inequality applied to the $L_1$ norm) that $s_h \leq s^* \leq 2\mathcal{S}(h)$. This means Laplace$(s_h/2\epsilon)$ has less variance than Laplace$(\mathcal{S}(h)/\epsilon)$ and so the quality of $p$-values for output perturbation under $2\epsilon$-mn-differential privacy are expected to be a lower bound on the quality for $\epsilon$-differential privacy. We note that typically,  $s_h$ can be much smaller than $s^*$ while $\mathcal{S}(h)\approx s^*/2$.

\subsection{Usage}\label{subsec:testbedusage}
We will use our permutation testbed to compare the likelihood ratio and $\chi^2$ statistics (with both input and output perturbation) to two other statistics that are rarely, if ever, used for independence testing in the non-private case, log-likelihood (LL) and absolute difference between actual count and expected count (Diff):
{\small
\begin{align}
LL &= -\Big[\sum_{i=1}^r \log(T[i,\bullet]!) + \sum_{j=1}^c \log(T[\bullet,j]!)
- \log(n!) - \sum_{i,j} \log(T[i,j]!)\Big]\label{eqn:lldef}\\
\text{Diff} &= \sum_{i,j} \left|T[i,j]-\frac{T[i,\bullet]T[\bullet,j]}{n}\right|\label{eqn:diffdef}
\end{align}
}

An important point of distinction is that for input perturbation, we will be using the noisy values $\widetilde{T}[i,j]$, $\widetilde{T}[i,\bullet]$ (equal to $\sum_j \widetilde{T}[i,j]$) and $\widetilde{T}[\bullet,j]$ to compute the statistics instead of the true values $T[i,j]$, $T[i,\bullet]$, $T[\bullet,j]$. This is because in the unrestricted case, the true values will not be available once the input has been perturbed.

In contrast, for output perturbation, we will use the true values of $T[i,\bullet], T[\bullet,j]$ to compute the statistics (since they are public in $\epsilon$-mn-differential privacy). This is because the output perturbation results would be a lower bound to the quality we should expect in the unrestricted case, and the sensitivity $s_h$ using this method is easier to compute (another advantage of this testbed!).

We now provide $s_h$ calculations for the output perturbation versions of the test statistics. In some cases, $s_h$ depends on the dimensions of $T$. 
\ConfOrTech{Due to space constraints, we present the most important $s_h$ results. The full results and proofs can be found in the long version \cite{wang2015differentially}.}{Proofs can be found in Appendix \ref{sec:sensitivityproofs}.} One important fact to note is that the Diff statistic has lowest $s_h$ and so, intuitively, is expected to perform well in the privacy setting.

\techreport{
\begin{restatable}{theorem}{cssensitivitytwo}
\label{thm:chi2sensitivity2by2}
 The $s_h$ value of the $\chi^2$-statistic for a $2 \times 2$ contingency tables is: 
 {\small
 \begin{align*}
  \max
  \begin{cases}
   C\left|n-2T[\cdot,2]T[1,\cdot]\right| & \text{if $T[1,\cdot] \leq T[\cdot,1]$, $T[2,\cdot] \geq T[\cdot,2]$} \\
   C\left|n-2T[\cdot,1]T[2,\cdot]\right| & \text{if $T[1,\cdot] > T[\cdot,1]$, $T[2,\cdot] < T[\cdot,2]$} \\
   C\left|n-2T[\cdot,1]T[1,\cdot]\right| & \text{if $T[1,\cdot] \leq T[\cdot,2]$, $T[2,\cdot] \geq T[\cdot,1]$} \\
   C\left|n-2T[\cdot,2]T[2,\cdot]\right| & \text{if $T[1,\cdot] > T[\cdot,2]$, $T[2,\cdot] < T[\cdot,1]$}
  \end{cases}
 \end{align*}
 }
 where $C=\frac{n^2}{T[\cdot,1]T[\cdot,2]T[1,\cdot]T[2,\cdot]}$.
\end{restatable}
The Proof of Theorem~\ref{thm:chi2sensitivity2by2} is provided in Appendix~\ref{subsec:chi2sensitivity2by2}.
Note that table margins are $O(n)$, so the constant $C$ in Theorem~\ref{thm:chi2sensitivity2by2} is $O(1/n^2)$ and so the $s_h$ value is $O(1)$ for $2\times 2$ tables. However, the chi-squared statistic does not grow with $n$ under the null hypothesis, so the noise can be a significant part of the output.

}

\begin{restatable}{theorem}{cssensitivity}
\label{thm:chi2sensitivity}
 The $s_h$ value of the $\chi^2$-statistic for an $r \times c$ table $T$ with $r \geq 3$, $c \geq 3$ is:
$$
  \max
  \begin{cases}
   \max\limits_{i_1,i_2,j_1,j_2}C'\Big|2(T[i_2,\cdot]T[\cdot,j_2]a+T[i_1,\cdot]T[\cdot,j_1]d)-
   (T[i_1,\cdot]+T[i_2,\cdot])(T[\cdot,j_1]+T[\cdot,j_2])\Big|/n
   \\
   \max\limits_{i_1,i_2,j_1,j_2}C'\Big|(T[i_1,\cdot]-T[i_2,\cdot])(T[\cdot,j_1]-T[\cdot,j_2])-
    2(T[i_2,\cdot]T[\cdot,j_1]b+T[i_1,\cdot]T[\cdot,j_2]c)\Big|/n
  \end{cases}
$$
 where $a=\min(T[i_1,\cdot],T[\cdot,j_1])$, $d=\min(T[i_2,\cdot],T[\cdot,j_2])$, $b=\min(T[i_1,\cdot],T[\cdot,j_2])-1$, $c=\min(T[i_2,\cdot],T[\cdot,j_1])-1$ and $C'=\frac{n^2}{T[\cdot,j_1]T[\cdot,j_2]T[i_1,\cdot]T[i_2,\cdot]}$.
\end{restatable}
\techreport{
The proof of Theorem~\ref{thm:chi2sensitivity} is provided in Appendix~\ref{subsec:chi2sensitivity}.
}

\techreport{
A simple analysis shows that the $s_h$ in Theorem \ref{thm:chi2sensitivity} is $O(1)$.

\begin{restatable}{theorem}{lrsensitivitytwo} 
\label{thm:lrsensitivity2by2}
 The $s_h$ value of the likelihood ratio statistic for $2 \times 2$ contingency tables is
\begin{align*}
  2 \times \max
  \begin{cases}
   \left|\log\frac{T[1,\cdot]^{T[1,\cdot]}}{(T[1,\cdot]-1)^{T[1,\cdot]-1}} + \log\frac{T[\cdot,2]^{T[\cdot,2]}}{(T[\cdot,2]-1)^{T[\cdot,2]-1}} +
   \log\frac{(T[2,\cdot]-T[\cdot,2])^{T[2,\cdot]-T[\cdot,2]}}{(T[2,\cdot]-T[\cdot,2]+1)^{T[2,\cdot]-T[\cdot,2]+1}}\right|
   \\
   \text{if $T[1,\cdot] \leq T[\cdot,1]$, $T[2,\cdot] \geq T[\cdot,2]$}
   \\
   \left|\log\frac{T[2,\cdot]^{T[2,\cdot]}}{(T[2,\cdot]-1)^{T[2,\cdot]-1}} + \log\frac{T[\cdot,1]^{T[\cdot,1]}}{(T[\cdot,1]-1)^{T[\cdot,1]-1}} +
   \log\frac{(T[\cdot,2]-T[2,\cdot])^{T[\cdot,2]-T[2,\cdot]}}{(T[\cdot,2]-T[2,\cdot]+1)^{T[\cdot,2]-T[2,\cdot]+1}}\right|
   \\
   \text{if $T[1,\cdot] > T[\cdot,1]$, $T[2,\cdot] < T[\cdot,2]$}
   \\
   \left|\log\frac{(T[1,\cdot]-1)^{T[1,\cdot]-1}}{T[1,\cdot]^{T[1,\cdot]}} + \log\frac{(T[\cdot,1]-1)^{T[\cdot,1]-1}}{T[\cdot,1]^{T[\cdot,1]}} +
   \log\frac{(T[2,\cdot]-T[\cdot,1]+1)^{T[2,\cdot]-T[\cdot,1]+1}}{(T[2,\cdot]-T[\cdot,1])^{T[2,\cdot]-T[\cdot,1]}}\right|
   \\
   \text{if $T[1,\cdot] \leq T[\cdot,2]$, $T[2,\cdot] \geq T[\cdot,1]$}
   \\
   \left|\log\frac{(T[2,\cdot]-1)^{T[2,\cdot]-1}}{T[2,\cdot]^{T[2,\cdot]}} + \log\frac{(T[\cdot,2]-1)^{T[\cdot,2]-1}}{T[\cdot,2]^{T[\cdot,2]}} +
   \log\frac{(T[1,\cdot]-T[\cdot,2]+1)^{T[1,\cdot]-T[\cdot,2]+1}}{(T[1,\cdot]-T[\cdot,2])^{T[1,\cdot]-T[\cdot,2]}}\right|
   \\
   \text{if $T[1,\cdot] > T[\cdot,2]$, $T[2,\cdot] < T[\cdot,1]$}
  \end{cases}
\end{align*}
\end{restatable}
The proof of Theorem~\ref{thm:lrsensitivity2by2} is provided in Appendix~\ref{subsec:lrsensitivity2by2}.
}

\begin{restatable}{theorem}{lrsensitivity} 
\label{thm:lrsensitivity}
 The $s_h$ of the likelihood ratio statistic LR on $r \times c$ ($r \geq 3$, $c \geq 3$) contingency tables is
 \begin{displaymath}
  2 \times \max
   \left\{\max\limits_{i_1,i_2,j_1,j_2}\left[\log\frac{a^a}{(a-1)^{a-1}} + \log\frac{d^d}{(d-1)^{d-1}}\right],
   \max\limits_{i_1,i_2,j_1,j_2}\left[\log\frac{(b+1)^{b+1}}{b^b} + \log\frac{(c+1)^{c+1}}{c^c}\right]\right\}
 \end{displaymath}
 where $a=\min(T[i_1,\cdot],T[\cdot,j_1])$, $d=\min(T[i_2,\cdot],T[\cdot,j_2])$, $b=\min(T[i_1,\cdot],T[\cdot,j_2])-1$, $c=\min(T[i_2,\cdot],T[\cdot,j_1])-1$.
\end{restatable}
\techreport{The proof of Theorem~\ref{thm:lrsensitivity} is provided in Appendix~\ref{subsec:lrsensitivity}.}
Thus $s_h = O(\log n)$ while the statistic itself does not grow with $n$ under the null hypothesis.


\techreport{
\begin{restatable}{theorem}{llsensitivitytwo}
\label{thm:llsensitivity2by2}
 The $s_h$ value of the log-likelihood statistic based on $2 \times 2$ contingency tables is
 \begin{equation*}
  \max
  \begin{cases}
   \left|\log(T[2,\cdot]-T[\cdot,2]+1)-\log T[1,\cdot]-\log T[\cdot,2]\right|
   \quad\text{if $T[1,\cdot] \leq T[\cdot,1]$, $T[2,\cdot] \geq T[\cdot,2]$}
   \\
   \left|\log(T[\cdot,2]-T[2,\cdot]+1)-\log T[2,\cdot]-\log T[\cdot,1]\right|
    \quad\text{if $T[1,\cdot] > T[\cdot,1]$, $T[2,\cdot] < T[\cdot,2]$}
   \\
   \left|\log T[1,\cdot]+\log T[\cdot,1]-\log(T[2,\cdot]-T[\cdot,1]+1)\right|
    \quad\text{if $T[1,\cdot] \leq T[\cdot,2]$, $T[2,\cdot] \geq T[\cdot,1]$}
   \\
   \left|\log T[2,\cdot]+\log T[\cdot,2]-\log(T[1,\cdot]-T[\cdot,2]+1)\right|
    \quad\text{if $T[1,\cdot] > T[\cdot,2]$, $T[2,\cdot] < T[\cdot,1]$}
  \end{cases}
 \end{equation*}
\end{restatable}
The proof of Theorem~\ref{thm:llsensitivity2by2} is provided in Appendix~\ref{subsec:llsensitivity2by2}.
}

\begin{restatable}{theorem}{llsensitivity}
\label{thm:llsensitivity}
 The $s_h$ value of the LL statistic (from Equation~\ref{eqn:lldef}) for $r \times c$ tables ($r \geq 3$, $c \geq 3$) is
 \begin{equation*}
  \max\left\{\max\limits_{i_1,i_2,j_1,j_2}\log\left[(b+1)(c+1)\right], \max\limits_{i_1,i_2,j_1,j_2}\log\left(ad\right)\right\}
 \end{equation*}
 where $a=\min(T[i_1,\cdot],T[\cdot,j_1])$, $d=\min(T[i_2,\cdot],T[\cdot,j_2])$, $b=\min(T[i_1,\cdot],T[\cdot,j_2])-1$, $c=\min(T[i_2,\cdot],T[\cdot,j_1])-1$.
\end{restatable}
\techreport{The proof of Theorem~\ref{thm:llsensitivity} is provided in Appendix~\ref{subsec:llsensitivity}.}
Here $s_h$ also grows logarithmically. 

\begin{restatable}{theorem}{absolutediffsensitivity}
\label{thm:absolutediffsensitivity}
 The $s_h$ value of the Diff statistic (from Equation~\ref{eqn:diffdef}) is equal to $4$.
\end{restatable}
\techreport{Proof of Theorem~\ref{thm:absolutediffsensitivity} is provided in \ConfOrTech{additional materials.}{Appendix~\ref{subsec:absolutediffsensitivity}.}}
The Diff statistic is the only one out of them that has a constant sensitivity. Clearly the value of the statistic should grow with $n$, even under the null distribution, so it should quickly  overwhelm the Laplace noise.

\subsection{Experiments for Permutation Testbed} \label{subsec:expperm}
Now we experiment with our permutation testbed to compare the $\chi^2$ and
likelihood ratio LR statistics to the non-traditional LL and Diff statistics
considered in Section \ref{subsec:testbedusage}. We compare the non-private
versions (evaluated on actual data) to the input perturbation (with suffix
``-in") and output perturbation (with suffix ``-out") versions. The method of \cite{uhler2013privacy,yu2014scalable} is denoted $\chi^2$-out and has lower noise requirements in the testbed than in general.

\techreport{ 
  \begin{figure}
    \centering
     \includegraphics[width=\textwidth]{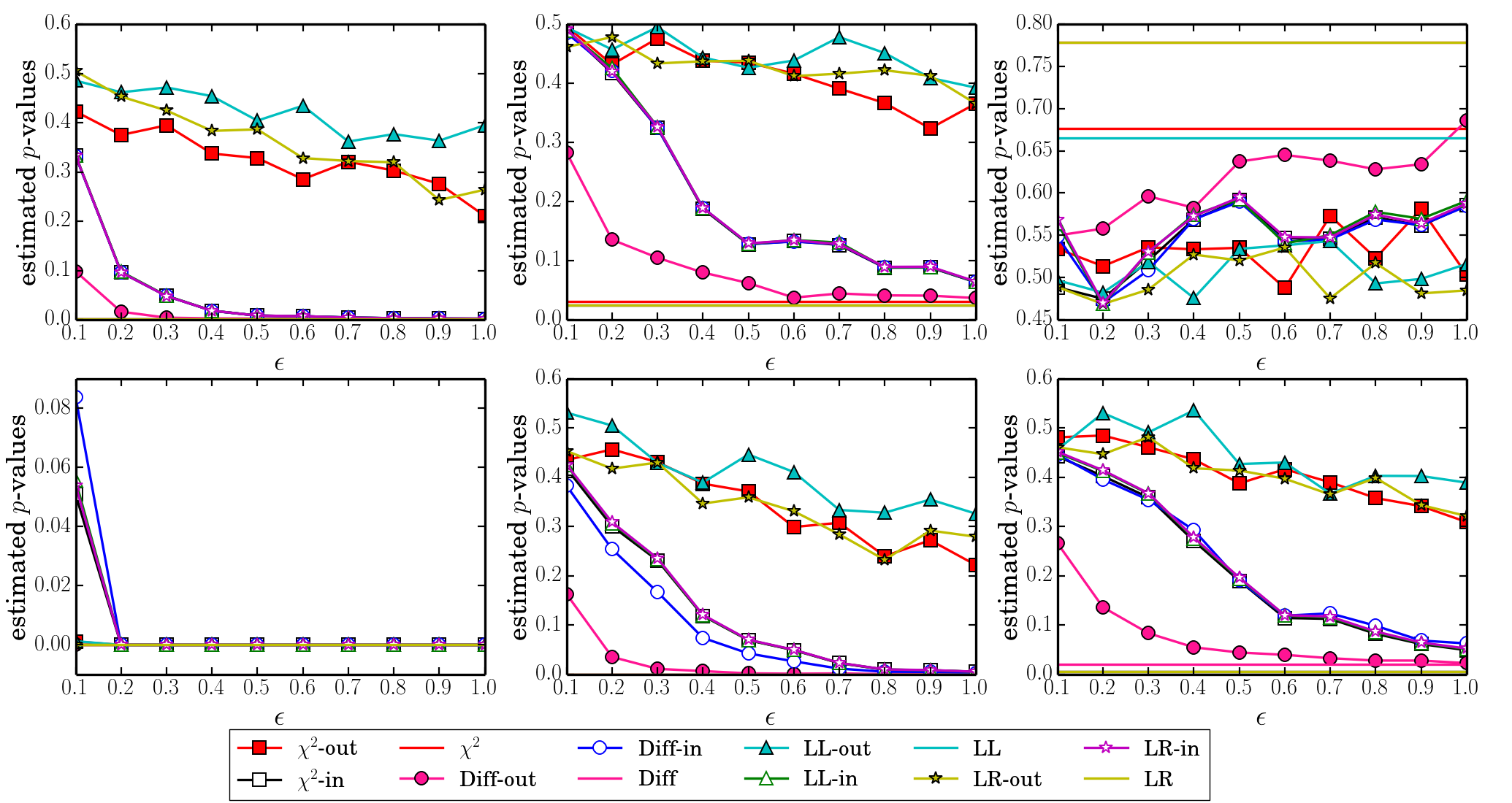}
     \vskip -14pt
     \caption{$P$-values of $\chi^2$ test, likelihood ratio test, log-likelihood
       test and \varchi{test} with input perturbation, output perturbation and
       no perturbation. Tables used: Czech \texttt{AD} with $n=1841$ (top left),
       Czech \texttt{BC} with $n=1841$ (top middle), Rochdale \texttt{AB} with
       $n=665$ (top right), Home Zone and Income Category with $n=2291$, $r=4$,
       $c=16$ (bottom left), Religion and Attitudes with $n=1055$, $r=c=3$
       (bottom middle), Religion and Education with $n=1055$, $r=c=3$ (bottom
       right).}
    \label{fig:realtableperm}
  \end{figure}

  Our first batch of results is shown in the top row of
  Fig.~\ref{fig:realtableperm} and illustrates three separate interesting
  phenomena. The left table has $p$-values that, in the non-private case, are
  generally considered highly significant. The output perturbation methods
  (including the perturbed $\chi^2$ statistic) generally have high variance and
  perform much worse than the input perturbation methods due to the amount of
  noise they require. The exception is the Diff statistic, whose output
  perturbation version requires the least amount of noise. The middle table has
  $p$-value that is often considered borderline for rejecting the null
  hypothesis. Again, we see the same pattern, with slightly more variance even
  for the input perturbation results. The reason for this is that the higher the
  non-private $p$-value, the smaller the value of the non-private test
  statistic. When that is small and when $n$ is small, the resulting value of
  the test statistic is easily dominated by the noise (creating large variance),
  but does not lead to unsupported rejections of the null hypothesis. This
  behavior is actually expected and desired. Even in the non-private case, when
  the null hypothesis is true, the $p$-values should be uniformly distributed
  while if the null hypothesis is false, the $p$-values should gravitate towards
  very small values.

}

The bottom row of Fig.~\ref{fig:realtableperm} shows typical results. Generally,
the output perturbation methods have high variance because of their noise
requirements. Meanwhile, input perturbation methods perform reasonably well even
in these tough noise scenarios. The exception is the output perturbation of the
Diff statistic, which is clearly the best and requires the least noise of all.
The Diff statistic grows with $n$, and so would not have an asymptotic
distribution in the unrestricted case, however, it is a promising starting point
upon which other privacy-aware statistics could be built.
\conferenceversion{Additional testbed experiments can be found in the full
  version \citep{wang2015differentially}.}
\section{Proof of Sensitivities}\label{sec:sensitivityproofs}
\subsection{Proof of Theorem~\ref{thm:chi2sensitivity2by2}}\label{subsec:chi2sensitivity2by2}

\cssensitivitytwo*
\begin{proof}\label{pf:chi2sensitivity2by2}
 From Definition~\ref{def:independence}, the $\chi^2$ statistic based on a $2 \times 2$ contingency table $T$ with fixed marginals $T[1,\cdot]$, $T[2,\cdot]$, $T[\cdot,1]$, $T[\cdot,2]$ is 
 $$\chi^2(T)=C/n \times (T[1,1]T[2,2]-T[1,2]T[2,1])^2$$
 
 From Definition~\ref{def:neighbors}, the neighboring contingency table $T'$ of $T$ has cell counts $T[1,1]-1, T[1,2]+1, T[2,1]+1, T[2,2]-1$. This implies the conditions $T[1,1] \geq 1$ and $T[2,2] \geq 1$\footnote{The case of incrementing $T[1,1]$,$T[2,2]$ and decrementing $T[1,2]$, $T[2,1]$ is symmetric because we can exchange $T$ and $T'$. This is also true for all other neighboring contingency tables with fixed  marginals}. From Definition~\ref{def:sensitivity}, the sensitivity equals
 \begin{align*}
 & \max_{T[1,1],T[1,2],T[2,1],T[2,2]} \left|\chi^2(T)-\chi^2(T')\right|
 \\
 = & \max_{T[1,1],T[1,2],T[2,1],T[2,2]} C\left|2T[1,1]T[2,2]-2T[1,2]T[2,1]-n\right|
 \end{align*}
 
 There are two ways to solve the above problem, that is, either maximize the formula inside the absolute value, or minimize it. Note we have the constraints $1\leq T[1,1]\leq \min(T[1,\cdot],T[\cdot,1])$, $0 \leq T[1,2] \leq \min(T[1,\cdot],T[\cdot,2])-1$, $0 \leq T[2,1] \leq \min(T[2,\cdot],T[\cdot,1])-1$, $1\leq T[2,2]\leq \min(T[2,\cdot],T[\cdot,2])$, and $T[1,\cdot]+T[2,\cdot]=T[\cdot,1]+T[\cdot,2]=n$. Since the marginals are fixed, we only have four variables. 
 
 In the first way, there are two cases.

 \begin{enumerate}
  \item If $T[1,\cdot] \leq T[\cdot,1]$, $T[2,\cdot] \geq T[\cdot,2]$
  
  It is easy to see $T[1,1]=T[1,\cdot]$, $T[2,2]=T[\cdot,2]$, $T[1,2]=0$ and $T[2,1]=T[2,\cdot]-T[\cdot,2]$ maximize the formula inside the absolute value. They give the result $C\left|n-2T[\cdot,2]T[1,\cdot]\right|$.
  
  \item If $T[1,\cdot] > T[\cdot,1]$, $T[2,\cdot] < T[\cdot,2]$
  
  It is easy to see $T[1,1]=T[\cdot,1]$, $T[2,2]=T[2,\cdot]$, $T[2,1]=0$ and $T[1,2]=T[\cdot,2]-T[2,\cdot]$ maximize the formula inside the absolute value. They give the result $C\left|n-2T[\cdot,1]T[2,\cdot]\right|$.
 \end{enumerate}
 
 In the second way, there are also two cases.

 \begin{enumerate}
  \item If $T[1,\cdot] \leq T[\cdot,2]$, $T[2,\cdot] \geq T[\cdot,1]$
  
  It is easy to see $T[1,2]=T[1,\cdot]-1$, $T[2,1]=T[\cdot,1]-1$, $T[1,1]=1$ and $T[2,2]=T[2,\cdot]-T[\cdot,1]+1$ minimize the formula inside the absolute value. They give the result $C\left|n-2T[\cdot,1]T[1,\cdot]\right|$.
  
  \item if $T[1,\cdot] > T[\cdot,2]$, $T[2,\cdot] < T[\cdot,1]$
  
  It is easy to see $T[1,2]=T[\cdot,2]-1$, $T[2,1]=T[2,\cdot]-1$, $T[1,1]=T[1,\cdot]-T[\cdot,2]+1$ and $T[2,2]=1$ minimize the formula inside the absolute value. They give the result $C\left|n-2T[\cdot,2]T[2,\cdot]\right|$.
 \end{enumerate}

 Therefore, the maximum value among all cases that apply to the marginals of table $T$ is its sensitivity with fixed marginals, which leads to the result in Theorem~\ref{thm:chi2sensitivity2by2}.
\end{proof}

\subsection{Proof of Theorem~\ref{thm:chi2sensitivity}}\label{subsec:chi2sensitivity}

\cssensitivity*
\begin{proof}\label{pf:chi2sensitivity}
 Let $T$ and $T'$ be neighboring contingency tables no smaller than $3 \times 3$ with fixed marginals. Suppose the four different entries between $T$ and $T'$ locate at the intersection of row $i_1, i_2$ and column $j_1, j_2$. We write $T[i_1,j_1], T[i_1,j_2], T[i_2,j_1], T[i_2,j_2]$ as $a, b, c, d$ respectively for short. The corresponding entries in $T'$ are then $a-1, b+1, c+1, d-1$. Note we have the conditions $a \geq 1$ and $d \geq 1$. By Definition~\ref{def:independence} and Definition~\ref{def:sensitivity}, the sensitivity of $\chi^2$-statistic can be computed by
 \begin{align*}
  \max_{T,T'} \left|\chi^2(T)-\chi^2(T')\right|
 = & \max_{T,T'} \left|\sum_{i,j}\frac{(T[i,j]-\frac{T[i,\cdot]T[\cdot,j]}{n})^2}{\frac{T[i,\cdot]T[\cdot,j]}{n}}-\sum_{i,j}\frac{(T'[i,j]-\frac{T'[i,\cdot]T'[\cdot,j]}{n})^2}{\frac{T'[i,\cdot]T'[\cdot,j]}{n}}\right| \\
 = & \max_{T} \left|\frac{2n(T[i_2,\cdot]T[\cdot,j_2]a-T[i_2,\cdot]T[\cdot,j_1]b-T[i_1,\cdot]T[\cdot,j_2]c)}{T[i_1,\cdot]T[i_2,\cdot]T[\cdot,j_1]T[\cdot,j_2]} \right. \\
 & \left.+ \frac{2nT[i_1,\cdot]T[\cdot,j_1]d-n(T[i_1,\cdot]+T[i_2,\cdot])(T[\cdot,j_1]+T[\cdot,j_2])}{T[i_1,\cdot]T[i_2,\cdot]T[\cdot,j_1]T[\cdot,j_2]}\right| 
 \end{align*} 
 
 Since all marginals are known, the only variables in the function are $a,b,c,d$. There are two ways to solve the problem.
 \begin{enumerate}
  \item Maximize the function inside the absolute value of the objective function by choosing large values for $a,d$ and small values for $b,c$.
  
    $b=c=0$ is the smallest possible values for them.
    $a=\min(T[i_1,\cdot],T[\cdot,j_1])$ and $d=\min(T[i_2,\cdot],T[\cdot,j_2])$
    are the largest possible values for them (respectively). These settings can
    form valid tables with dimensions at least $3 \times 3$. Plug them
    into the objective gives
    $C'|2(T[i_2,\cdot]T[\cdot,j_2]a+T[i_1,\cdot]T[\cdot,j_1]d)-(T[i_1,\cdot]+T[i_2,\cdot])(T[\cdot,j_1]+T[\cdot,j_2])|/n$.
    Then the indices $i_1,i_2,j_1,j_2$ maximizing it should be chosen.
  
  \item Minimize the function inside the absolute value of the objective function by choosing large values for $b,c$ and small values for $a,d$.
  
  $a=d=1$ is the smallest possible values for them. $b=\min(T[i_1,\cdot],T[\cdot,j_2])-1$ and $c=\min(T[i_2,\cdot],T[\cdot,j_1])-1$ are the largest possible values for them (respectively). These settings can also form a valid table. Plugging them back to the objective function gives $C'|(T[i_1,\cdot]-T[i_2,\cdot])(T[\cdot,j_1]-T[\cdot,j_2])-2(T[i_2,\cdot]T[\cdot,j_1]b+T[i_1,\cdot]T[\cdot,j_2]c)|/n$. Then the indices $i_1,i_2,j_1,j_2$ maximizing it should be chosen.
 \end{enumerate}
 
 The sensitivity should be the larger one computed from the above two cases, which gives the result in Theorem~\ref{thm:chi2sensitivity}.

\end{proof}
 
\subsection{Proof of Theorem~\ref{thm:lrsensitivity2by2}}\label{subsec:lrsensitivity2by2}

\lrsensitivitytwo*
\begin{proof}\label{pf:lrsensitivity2by2}
 Suppose $2 \times 2$ contingency table $T$ has fixed marginals $T[1,\cdot]$, $T[2,\cdot]$, $T[\cdot,1]$, $T[\cdot,2]$. From Definition~\ref{def:neighbors}, the neighboring contingency table $T'$ of $T$ has cell counts $T[1,1]-1, T[1,2]+1, T[2,1]+1, T[2,2]-1$. This implies the conditions $T[1,1] \geq 1$ and $T[2,2] \geq 1$. We also have the conditions $T[1,\cdot]+T[2,\cdot]=T[\cdot,1]+T[\cdot,2]=n$. From  Definition~\ref{def:independence} and \ref{def:sensitivity}, the sensitivity equals

 \begin{align*}
  & \max_{T,T'} 2\left|\sum_{i,j}T[i,j]\log\frac{nT[i,j]}{T[i,\cdot]T[\cdot,j]} - \sum_{i,j}T'[i,j]\log\frac{nT'[i,j]}{T'[i,\cdot]T'[\cdot,j]}\right| \\
  = & \max_{T[1,1],T[1,2],T[2,1],T[2,2]}2\left|\log\frac{T[1,1]^{T[1,1]}}{(T[1,1]-1)^{T[1,1]-1}} + \log\frac{T[1,2]^{T[1,2]}}{(T[1,2]+1)^{T[1,2]+1}} \right.\\
&+\left. \log\frac{T[2,1]^{T[2,1]}}{(T[2,1]+1)^{T[2,1]+1}} + \log\frac{T[2,2]^{T[2,2]}}{(T[2,2]-1)^{T[2,2]-1}} \right|
 \end{align*}
 There are two ways to solve the above problem, that is, either maximize the formula inside the absolute value of the above objective function, or minimize it. Note we have the constraints $1\leq T[1,1]\leq \min(T[1,\cdot],T[\cdot,1])$, $0 \leq T[1,2] \leq \min(T[1,\cdot],T[\cdot,2])-1$, $0 \leq T[2,1] \leq \min(T[2,\cdot],T[\cdot,1])-1$, $1\leq T[2,2]\leq \min(T[2,\cdot],T[\cdot,2])$. 
 
 The derivative of $\log\frac{T[1,1]^{T[1,1]}}{(T[1,1]-1)^{T[1,1]-1}}$ with respect to $T[1,1]$ is $\log\frac{T[1,1]}{T[1,1]-1}$. When $T[1,1]>1$, the term and its derivative are both positive; when $T[1,1]=1$, the term equals $0$. The last term $\log\frac{T[2,2]^{T[2,2]}}{(T[2,2]-1)^{T[2,2]-1}}$ has exactly the same analysis, and so does $T[2,2]$. For the term $\log\frac{T[1,2]^{T[1,2]}}{(T[1,2]+1)^{T[1,2]+1}}$, its derivative with respect to $T[1,2]$ equals $\log\frac{T[1,2]}{T[1,2]+1}$. Both the term and its derivative are negative when $T[1,2]>0$. When $T[1,2]=0$, the term equals $0$. Similarly, we can apply the same analysis to $T[2,1]$ as what we do to $T[1,2]$. We use this derivative analysis for the two ways of solving the problem.
 
 In the first way, there are two cases.

 \begin{enumerate}
  \item If $T[1,\cdot] \leq T[\cdot,1]$, $T[2,\cdot] \geq T[\cdot,2]$
  
  It is easy to see $T[1,1]=T[1,\cdot]$, $T[2,2]=T[\cdot,2]$, $T[1,2]=0$ and $T[2,1]=T[2,\cdot]-T[\cdot,2]$ maximize the formula inside the absolute value, which leads to $\left|\log\frac{T[1,\cdot]^{T[1,\cdot]}}{(T[1,\cdot]-1)^{T[1,\cdot]-1}} + \log\frac{T[\cdot,2]^{T[\cdot,2]}}{(T[\cdot,2]-1)^{T[\cdot,2]-1}} + \log\frac{(T[2,\cdot]-T[\cdot,2])^{T[2,\cdot]-T[\cdot,2]}}{(T[2,\cdot]-T[\cdot,2]+1)^{T[2,\cdot]-T[\cdot,2]+1}}\right|$.
  
  \item If $T[1,\cdot] > T[\cdot,1]$, $T[2,\cdot] < T[\cdot,2]$
  
  It is easy to see $T[1,1]=T[\cdot,1]$, $T[2,2]=T[2,\cdot]$, $T[2,1]=0$ and $T[1,2]=T[\cdot,2]-T[2,\cdot]$ maximize the formula inside the absolute value, which leads to $\left|\log\frac{T[2,\cdot]^{T[2,\cdot]}}{(T[2,\cdot]-1)^{T[2,\cdot]-1}} + \log\frac{T[\cdot,1]^{T[\cdot,1]}}{(T[\cdot,1]-1)^{T[\cdot,1]-1}} + \log\frac{(T[\cdot,2]-T[2,\cdot])^{T[\cdot,2]-T[2,\cdot]}}{(T[\cdot,2]-T[2,\cdot]+1)^{T[\cdot,2]-T[2,\cdot]+1}}\right|$.
 \end{enumerate}
 
 In the second way, there are also two cases.

 \begin{enumerate}
  \item If $T[1,\cdot] \leq T[\cdot,2]$, $T[2,\cdot] \geq T[\cdot,1]$
  
  It is easy to see $T[1,2]=T[1,\cdot]-1$, $T[2,1]=T[\cdot,1]-1$, $T[1,1]=1$ and $T[2,2]=T[2,\cdot]-T[\cdot,1]+1$ minimize the formula inside the absolute value, which leads to $\left|\log\frac{(T[1,\cdot]-1)^{T[1,\cdot]-1}}{T[1,\cdot]^{T[1,\cdot]}} + \log\frac{(T[\cdot,1]-1)^{T[\cdot,1]-1}}{T[\cdot,1]^{T[\cdot,1]}} + \log\frac{(T[2,\cdot]-T[\cdot,1]+1)^{T[2,\cdot]-T[\cdot,1]+1}}{(T[2,\cdot]-T[\cdot,1])^{T[2,\cdot]-T[\cdot,1]}}\right|$.
  
  \item if $T[1,\cdot] > T[\cdot,2]$, $T[2,\cdot] < T[\cdot,1]$
  
  It is easy to see $T[1,2]=T[\cdot,2]-1$, $T[2,1]=T[2,\cdot]-1$, $T[1,1]=T[1,\cdot]-T[\cdot,2]+1$ and $T[2,2]=1$ minimize the formula inside the absolute value, which leads to $\left|\log\frac{(T[2,\cdot]-1)^{T[2,\cdot]-1}}{T[2,\cdot]^{T[2,\cdot]}} + \log\frac{(T[\cdot,2]-1)^{T[\cdot,2]-1}}{T[\cdot,2]^{T[\cdot,2]}} + \log\frac{(T[1,\cdot]-T[\cdot,2]+1)^{T[1,\cdot]-T[\cdot,2]+1}}{(T[1,\cdot]-T[\cdot,2])^{T[1,\cdot]-T[\cdot,2]}}\right|$.
 \end{enumerate}
 
 Therefore, the maximum value among all cases that apply to the marginals of table $T$ is its sensitivity with fixed marginals, which leads to the result in Theorem~\ref{thm:lrsensitivity2by2}.
\end{proof}

\subsection{Proof of Theorem~\ref{thm:lrsensitivity}}\label{subsec:lrsensitivity}

\lrsensitivity*
\begin{proof}\label{pf:lrsensitivity}
 Let $T$ and $T'$ be neighboring contingency tables no smaller than $3 \times 3$ with fixed marginals. Suppose the four different entries between $T$ and $T'$ locate at the intersection of row $i_1, i_2$ and column $j_1, j_2$. We write $T[i_1,j_1], T[i_1,j_2], T[i_2,j_1], T[i_2,j_2]$ as $a, b, c, d$ respectively for short. The corresponding entries in $T'$ are then $a-1, b+1, c+1, d-1$. Note we have the conditions $a \geq 1$ and $d \geq 1$. By Definition~\ref{def:independence} and Definition~\ref{def:sensitivity}, the sensitivity of likelihood ratio statistic can be computed by
 \begin{align*}
  & \max_{T,T'} 2\left|\sum_{i,j}T[i,j]\log\frac{nT[i,j]}{T[i,\cdot]T[\cdot,j]} - \sum_{i,j}T'[i,j]\log\frac{nT'[i,j]}{T'[i,\cdot]T'[\cdot,j]}\right| \\
  = & \max_{a,b,c,d}2\left|\log\frac{a^a}{(a-1)^{a-1}} + \log\frac{b^b}{(b+1)^{b+1}} + \log\frac{c^c}{(c+1)^{c+1}} + \log\frac{d^d}{(d-1)^{d-1}} \right|
 \end{align*}
 The objective function only contains variables $a,b,c,d$. So, following from the proof for Theorem~\ref{thm:lrsensitivity2by2}, we do the same thing. That is, either maximize or minimize the formula inside the absolute value from the objective function.
 
 In the first case, we choose $b=c=0$, $a=\min(T[i_1,\cdot],T[\cdot,j_1])$, $d=\min(T[i_2,\cdot],T[\cdot,j_2])$, which gives the result $\log\frac{a^a}{(a-1)^{a-1}} + \log\frac{d^d}{(d-1)^{d-1}}$. Next, we find the indices $i_1,i_2,j_1,j_2$ which maximizes it.

 In the second case, we choose $a=d=1$, $b=\min(T[i_1,\cdot],T[\cdot,j_2])-1$, $c=\min(T[i_2,\cdot],T[\cdot,j_1])-1$, which gives the result $\log\frac{(b+1)^{b+1}}{b^b} + \log\frac{(c+1)^{c+1}}{c^c}$. Next, we find the indices $i_1,i_2,j_1,j_2$ which maximizes it.
 
 The sensitivity is the larger of the above two cases, which leads to Theorem~\ref{thm:lrsensitivity}.
\end{proof}

\subsection{Proof of Theorem~\ref{thm:llsensitivity2by2}}\label{subsec:llsensitivity2by2}

\llsensitivitytwo*
\begin{proof}\label{pf:llsensitivity2by2}
 Suppose $2 \times 2$ contingency table $T$ has fixed marginals $T[1,\cdot]$, $T[2,\cdot]$, $T[\cdot,1]$, $T[\cdot,2]$. From Definition~\ref{def:neighbors}, the neighboring contingency table $T'$ of $T$ has cell counts $T[1,1]-1, T[1,2]+1, T[2,1]+1, T[2,2]-1$. This implies the conditions $T[1,1] \geq 1$ and $T[2,2] \geq 1$. We also have the conditions $T[1,\cdot]+T[2,\cdot]=T[\cdot,1]+T[\cdot,2]=n$. From  Equation~\ref{eqn:lldef} and \ref{def:sensitivity}, the sensitivity equals
 
 \begin{align*}
  & \max_{T,T'} \left|\sum_{i}\log(T[i,\cdot]!) + \sum_{j}\log(T[\cdot,j]!) - \sum_{i,j}\log(T[i,j]!)\right. \\
  & \left. - \left[\sum_{i}\log(T'[i,\cdot]!) + \sum_{j}\log(T'[\cdot,j]!) - \sum_{i,j}\log(T'[i,j]!) \right] \right| \\
  = & \max_{T[1,1],T[1,2],T[2,1],T[2,2]} \left|-\log T[1,1]! - \log T[1,2]! - \log T[2,1]! - \log T[2,2]! + \log(T[1,1]-1)!\right. \\
  & \left.  + \log(T[1,2]+1)! + \log(T[2,1]+1)! + \log(T[2,2]-1)!\right| \\
  = & \max_{T[1,1],T[1,2],T[2,1],T[2,2]} \left|-\log T[1,1] + \log(T[1,2]+1) + \log(T[2,1]+1) - \log T[2,2] \right|
 \end{align*}
 
 We solve the objective function by either minimizing the formula inside the absolute value of the objective function or maximizing it. 

 In the first way, there are two cases.

 \begin{enumerate}
  \item If $T[1,\cdot] \leq T[\cdot,1]$, $T[2,\cdot] \geq T[\cdot,2]$
  
  It is easy to see $T[1,1]=T[1,\cdot]$, $T[2,2]=T[\cdot,2]$, $T[1,2]=0$ and $T[2,1]=T[2,\cdot]-T[\cdot,2]$ minimize it, which leads to $\left|\log(T[2,\cdot]-T[\cdot,2]+1)-\log T[1,\cdot]-\log T[\cdot,2]\right|$.
  
  \item If $T[1,\cdot] > T[\cdot,1]$, $T[2,\cdot] < T[\cdot,2]$
  
  It is easy to see $T[1,1]=T[\cdot,1]$, $T[2,2]=T[2,\cdot]$, $T[2,1]=0$ and $T[1,2]=T[\cdot,2]-T[2,\cdot]$ minimize it, which leads to $\left|\log(T[\cdot,2]-T[2,\cdot]+1)-\log T[2,\cdot]-\log T[\cdot,1]\right|$.
 \end{enumerate}
 
 In the second way, there are also two cases.

 \begin{enumerate}
  \item If $T[1,\cdot] \leq T[\cdot,2]$, $T[2,\cdot] \geq T[\cdot,1]$
  
  It is easy to see $T[1,2]=T[1,\cdot]-1$, $T[2,1]=T[\cdot,1]-1$, $T[1,1]=1$ and $T[2,2]=T[2,\cdot]-T[\cdot,1]+1$ maximize it, which leads to $\left|\log T[1,\cdot]+\log T[\cdot,1]-\log(T[2,\cdot]-T[\cdot,1]+1)\right|$.
  
  \item if $T[1,\cdot] > T[\cdot,2]$, $T[2,\cdot] < T[\cdot,1]$
  
  It is easy to see $T[1,2]=T[\cdot,2]-1$, $T[2,1]=T[2,\cdot]-1$, $T[1,1]=T[1,\cdot]-T[\cdot,2]+1$ and $T[2,2]=1$ maximize it, which leads to $\left|\log T[2,\cdot]+\log T[\cdot,2]-\log(T[1,\cdot]-T[\cdot,2]+1)\right|$.
 \end{enumerate}
 
 Therefore, the maximum value among all cases that apply to the marginals of table $T$ is its sensitivity with fixed marginals, which leads to the result in Theorem~\ref{thm:llsensitivity2by2}.
\end{proof}

\subsection{Proof of Theorem~\ref{thm:llsensitivity}}\label{subsec:llsensitivity}

\llsensitivity*
\begin{proof}\label{pf:llsensitivity}
 Let $T$ and $T'$ be neighboring contingency tables no smaller than $3 \times 3$ with fixed marginals. Suppose the four different entries between $T$ and $T'$ locate at the intersection of row $i_1, i_2$ and column $j_1, j_2$. We write $T[i_1,j_1], T[i_1,j_2], T[i_2,j_1], T[i_2,j_2]$ as $a, b, c, d$ respectively for short. The corresponding entries in $T'$ are then $a-1, b+1, c+1, d-1$. Note we have the conditions $a \geq 1$ and $d \geq 1$. By Equation~\ref{eqn:lldef} and Definition~\ref{def:sensitivity}, the sensitivity of log-likelihood statistic can be computed by
 \begin{align*}
  & \max_{T,T'} \left|\sum_{i}\log(T[i,\cdot]!) + \sum_{j}\log(T[\cdot,j]!) - \sum_{i,j}\log(T[i,j]!)\right. \\
  & \left. - \left[\sum_{i}\log(T'[i,\cdot]!) + \sum_{j}\log(T'[\cdot,j]!) - \sum_{i,j}\log(T'[i,j]!) \right] \right| \\
  = & \max_{a,b,c,d} \left|-\log a! - \log b! - \log c! - \log d! + \log(a-1)!  + \log(b+1)! + \log(c+1)! + \log(d-1)!\right| \\
  = & \max_{a,b,c,d} \left|-\log a + \log(b+1) + \log(c+1) - \log d \right|
 \end{align*}
 
 It is easy to see the objective function is maximized by either $a=d=1$, $b=\min(T[i_1,\cdot],T[\cdot,j_2])-1$, $c=\min(T[i_2,\cdot],T[\cdot,j_1])-1$ or $b=c=0$, $a=\min(T[i_1,\cdot],T[\cdot,j_1])$, $d=\min(T[i_2,\cdot],T[\cdot,j_2])$, which leads to $\log\left[(b+1)(c+1)\right]$ and $\log\left(ad\right)$ respectively. Next, we find the indices $i_1,i_2,j_1,j_2$ that maximizes them separately.
 
 The sensitivity is the larger from the two cases, which leads to Theorem~\ref{thm:llsensitivity}.
\end{proof}

%

\subsection{Proof of Theorem~\ref{thm:absolutediffsensitivity}}\label{subsec:absolutediffsensitivity}

\absolutediffsensitivity*
\begin{proof}\label{pf:absolutediffsensitivity}
 Suppose $T$ and $T'$ are marginal-neighbors defined in Definition~\ref{def:neighbors}. So, there are indices $i_1,i_2,j_1,j_2$ such that $T'[i_1,j_1]=T[i_1,j_1]-1$, $T'[i_1,j_2]=T[i_1,j_2]+1$, $T'[i_2,j_1]=T[i_2,j_1]+1$, $T'[i_2,j_2]=T[i_2,j_2]-1$. Also, recall that both tables have the same marginals. By Equation~\ref{eqn:diffdef} and \ref{def:sensitivity}, the sensitivity for \varchi{statistic} equals
 \begin{align*}
  & \max_{T,T'} \left|\sum_{i,j} \left|T[i,j]-\frac{T[i,\cdot]T[\cdot,j]}{n}\right| - \sum_{i,j} \left|T'[i,j]-\frac{T'[i,\cdot]T'[\cdot,j]}{n}\right|\right|
  \\
  = & \max_{T,T'} \Biggl|\left|T[i_1,j_1]-\frac{T[i_1,\cdot]T[\cdot,j_1]}{n}\right| + \left|T[i_1,j_2]-\frac{T[i_1,\cdot]T[\cdot,j_2]}{n}\right| + \left|T[i_2,j_1]-\frac{T[i_2,\cdot]T[\cdot,j_1]}{n}\right| \\
 & + \left|T[i_2,j_2]-\frac{T[i_2,\cdot]T[\cdot,j_2]}{n}\right| - \left|T'[i_1,j_1]-\frac{T'[i_1,\cdot]T'[\cdot,j_1]}{n}\right| - \left|T'[i_1,j_2]-\frac{T'[i_1,\cdot]T'[\cdot,j_2]}{n}\right|\\
& - \left|T'[i_2,j_1]-\frac{T'[i_2,\cdot]T'[\cdot,j_1]}{n}\right| - \left|T'[i_2,j_2]-\frac{T'[i_2,\cdot]T'[\cdot,j_2]}{n}\right|\Biggr|
  \\
  \leq & \left|T[i_1,j_1]-T'[i_1,j_1]\right| + \left|T[i_1,j_2]-T'[i_1,j_2]\right| + \left|T[i_2,j_1]-T'[i_2,j_1]\right| + \left|T[i_2,j_2]-T'[i_2,j_2]\right|
  \\
  = & 4
 \end{align*}
 
 That is, the sensitivity of the \varchi{statistic} based on contingency tables with fixed marginals is $4$.
\end{proof}

}
\end{document}